\documentclass[10pt, conference, letterpaper]{IEEEtran}

\ifCLASSOPTIONcompsoc
  \usepackage[nocompress]{cite}
\else
  \usepackage{cite}
\fi
\usepackage{cite,url}
\usepackage{booktabs}
\usepackage{amsmath,array,amssymb,bm,dsfont,bbm,mathtools}
\usepackage{algorithmic}
\usepackage{amsthm,nccmath,lipsum,cuted}
\usepackage[dvips]{color}
\usepackage{xcolor}

\usepackage{graphicx}
\usepackage[lined,boxed,commentsnumbered, ruled ]{algorithm2e}
\usepackage{subcaption}
\usepackage{cleveref}
\usepackage{comment}

\usepackage{amsfonts}
\usepackage{ragged2e}
\newtheorem{theorem}{Theorem}
\newtheorem{lemma}{Lemma}

\newtheorem{corollary}{Corollary}

\usepackage{verbatim}
\usepackage{comment}
\usepackage{float}
\usepackage{xspace}
\usepackage[footnote,draft,danish,silent,nomargin]{fixme}

\def\algo{AMR$^2$\xspace}
\def\blalgo{Greedy-RRA\xspace}
\def\subilp{sub-ILP\xspace}

\def\P{{\mathcal P}}  %%%   appears in many equations  Prob
\def\INTGR{{\mathds Z}}   % Integer numbers
\def\PI{{\mathcal{P}_\text{I}}}

\DeclareMathOperator*{\argmax}{arg\,max}

\newcommand{\floor}[1]{\left\lfloor #1 \right\rfloor}

\graphicspath{{./figs/}}

\makeatletter
\newcommand{\pushright}[1]{\ifmeasuring@#1\else\omit\hfill$\displaystyle#1$\fi\ignorespaces}
\newcommand{\mathleft}{\@fleqntrue\@mathmargin0pt}
\newcommand{\mathcenter}{\@fleqnfalse}
\makeatother
\begin{document}

%\title{Job Scheduling for Improving Inference Accuracy in Computation Offloading Under a Deadline Constraint}
%\title{Job Scheduling for Maximizing Inference Accuracy on Edge Device Under a Deadline Constraint}
\title{Offloading Algorithms for Maximizing Inference Accuracy on Edge Device Under a Time Constraint}
\author{
Andrea Fresa and Jaya Prakash Champati\\
IMDEA Networks Institute, Madrid, Spain\\
	E-mail:$\{$andrea.fresa,jaya.champati$\}$@imdea.org
}
\maketitle

%\title{Scheduling Algorithms for Improving Edge Device Inference with Scalable Models and Job Offloading}

\begin{abstract}
With the emergence of edge computing, the problem of offloading jobs between an Edge Device (ED) and an Edge Server (ES) received significant attention in the past. Motivated by the fact that an increasing number of applications are using Machine Learning (ML) inference from the data samples collected at the EDs, we study the problem of offloading \textit{inference jobs}
%, where an inference job refers to the execution of a pre-trained ML model on a data sample, 
by considering the following novel aspects: 1) in contrast to a typical computational job, the processing time of an inference job depends on the size of the ML model, and 2) recently proposed Deep Neural Networks (DNNs) for resource-constrained devices provide the choice of scaling down the model size by trading off the inference accuracy. Considering that multiple ML models are available at the ED, and a powerful ML model is available at the ES, we formulate a general assignment problem with the objective of maximizing the total inference accuracy of $n$ data samples at the ED subject to a time constraint $T$ on the makespan. Noting that the problem is NP-hard, we propose an approximation algorithm Accuracy Maximization using LP-Relaxation and Rounding (\algo), and prove that it results in a makespan at most $2T$, and achieves a total accuracy that is lower by a small constant from the optimal total accuracy. Further, if the data samples are identical we propose Accuracy Maximization using Dynamic Programming (AMDP), an optimal pseudo-polynomial time algorithm. As proof of concept, we implemented \algo on a Raspberry Pi, equipped with MobileNets, that is connected to a server equipped with ResNet, and studied the total accuracy and makespan performance of \algo for image classification.
\end{abstract}

\section{Introduction}\label{sec:intro}
%%%%% Problem context and Importance %%%%%%%%%

%What if we start with ML aspect and then present the edge computing aspect?
%Should we mention Edge Intelligence?

Edge computing is seen as a key component of future networks that augments the computation, memory, and battery limitations of Edge Devices (EDs) (e.g., IoT devices, mobile phones, etc.), by allowing the devices to offload computational jobs to nearby Edge Servers (ESs)~\cite{Shi2016}. Since the \textit{offloading decision}, i.e., which jobs to offload, is the key to minimizing the execution delay of the jobs and/or the energy consumption at the ED, it received significant attention in the past~\cite{Mach2017}.
%However, little attention has been given to the offloading problem in the context of growing number of applications which require Machine Learning (ML) inference on the data collected by the ED.
%There is significant thrust
%Edge computing is seen as a key component of future networks as it enables job/data processing on edge servers, such as cloudlets~\cite{Satyanarayanan2009}, that are in close proximity to the edge devices, e.g., smartphones, IoT devices, etc., which generate the jobs/data~\cite{Shi2016}. 
%A key decision at an Edge Device (ED) is whether to  computation offloading from an Edge Device (ED) to an Edge Server (ES) 
%Recently, increasing number of applications use Machine Learning (ML) inference from pre-trained ML models, such as Deep Neural Networks (DNNs), deployed on the ED. This is made possible by a large body of ML research on DNN models that trades inference accuracy for reduced computation and storage requirements, e.g., see \cite{Howard2017,Deng2020}.  
Recently, an increasing number of applications are using Machine Learning (ML) inference from the data samples collected at the EDs, and  
there is a major thrust for deploying pre-trained Deep Neural Networks (DNNs) on the EDs as this has, among other advantages, reduced latency.
%the advantage would result in reduced latency and/or power consumption, and preserves privacy. 
Thanks to the development of DNN models with reduced computation and storage requirements, possibly with reduced inference accuracy, and the advancements in the hardware of EDs\cite{Deng2020}, ML frameworks such as Tensorflow Lite~\cite{TFLite} and PyTorch Mobile~\cite{PyTorch} are now able to support the deployment of DNNs on EDs. 
%that trade inference accuracy for reduced computation and storage requirements \cite{Deng2020}. Also, frameworks, such as Tensorflow Lite~\cite{TFLite} and Caffe2, for running DNNs on EDs have been developed. 
In this context, we study the offloading decision between an ED and an ES for the \textit{inference jobs}, where an inference job refers to the execution of a pre-trained ML model on a data sample.
%However, very few existing  offloading decision for the jobs involving Machine Learning (ML) inference on the data collected at the ED.

In comparison to the fixed processing time requirement of a generic computational job (typically represented by a directed acyclic task graph), the processing time requirement of an inference job depends on the ML model size: a larger model size results in longer processing time and may provide higher inference accuracy. For example, on Pixel 3 smartphone, ResNet~\cite{He2016} has size $178$ MB, requires $526$ ms, and provides $76.8$\% accuracy (Top-1 accuracy) for the ImageNet dataset~\cite{Imagenet2009}, while the smallest DNN model of MobileNet~\cite{Sandler2018} has size $1.9$ MB, requires $1.2$ ms, but provides $41.4$\% accuracy~\cite{TFLite}. Furthermore, recently developed DNNs for EDs allow for scaling the model size by simply setting a few hypeparameters (cf. \cite{teerapittayanon2016,Sandler2018,Han2019}),  enabling the EDs to choose between multiple model sizes. However, as we explain in Section~\ref{sec:related}, the offloading decision for inference jobs considering the above novel aspects has received little attention in the literature.
\begin{figure}[t!]
    \centering
    \includegraphics[width=2.7in]{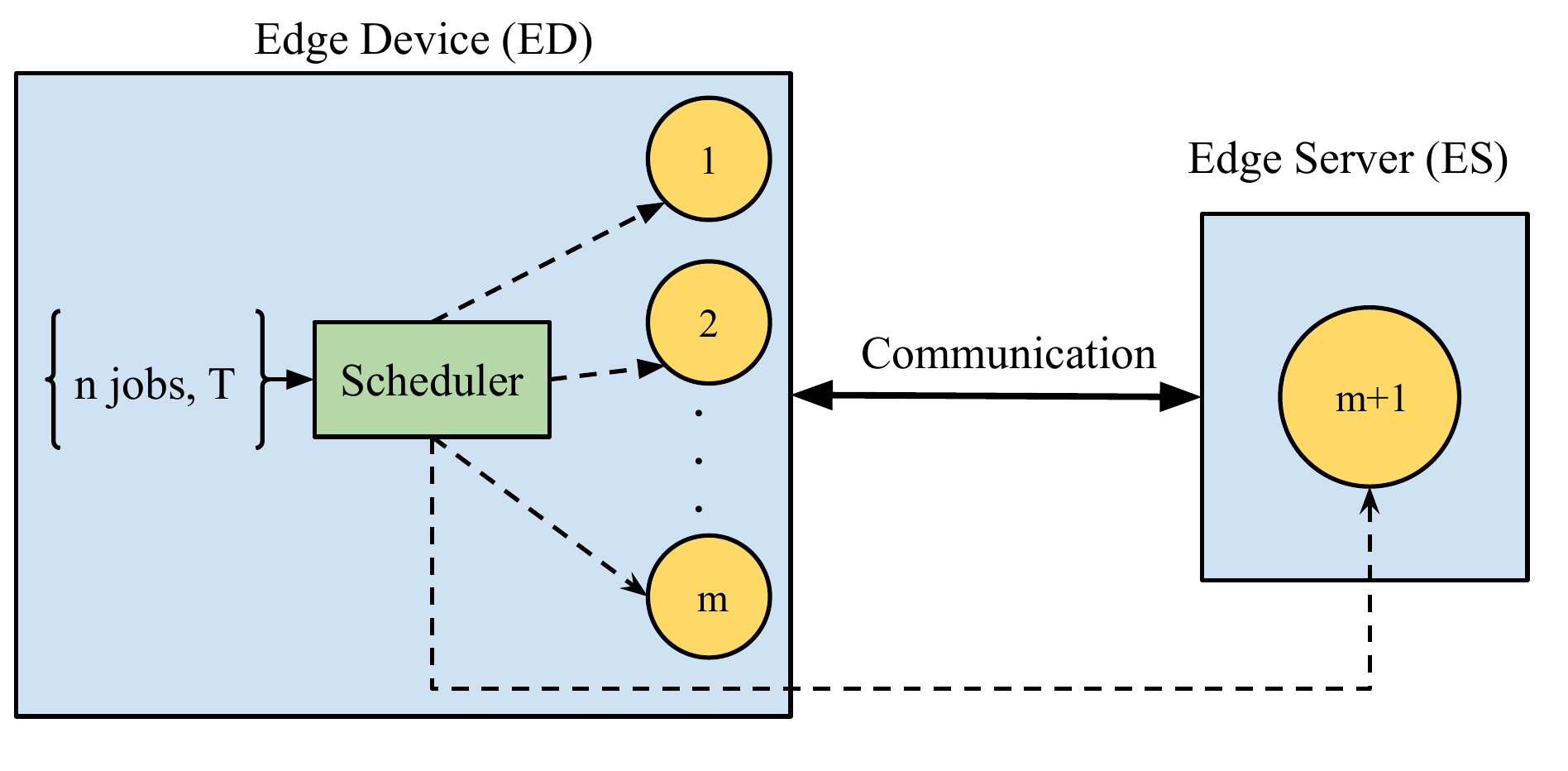}
    \caption{Scheduling inference jobs between an ED and an ES.}
    \vspace{-0.5cm}
    \label{fig:system}
\end{figure}
Taking into account the novel aspects for inference jobs, or simply \textit{jobs} in the sequel, we consider the system in Figure~\ref{fig:system}, where the ED has $m$ ML models to choose from, and the ES is equipped with a state-of-the-art ML model for a given application. 
%where an ML model may represent an instantiation of a DNN with a different hyperparameter value  of the same DNN or   
Consider that $n$ jobs (corresponding to $n$ data samples) are available at the ED. It may offload them all to the ES to maximize the inference accuracy. However, offloading each job incurs a communication time to upload the data sample in addition to the processing time at the ES. This may result in a large \textit{makespan}, i.e., the total time to finish all the jobs. On the other hand, executing all the jobs on the smallest ML model at the ED may result not only in a smaller makespan, but also the lowest inference accuracy. Thus, a scheduler at the ED needs to strike a trade-off between the accuracy and the makespan. Toward this end, we formulate the following problem: \textit{given $n$ data samples at time zero, find a schedule that offloads a partition of the jobs to the ES and assigns the remaining jobs to $m$  models on the ED, such that the total accuracy is maximized and the makespan is within a time constraint $T$}. 
Solution to this problem will be beneficial to applications such as Google Photos where a set of photos selected by a user need to be classified into multiple categories in real time. Also, the problem has relevance to applications which do periodic scheduling, i.e., the ED periodically collects all the data samples arrived in a time period $T$ and aims to finish their processing within the next time period $T$.
%Recently, offloading inference jobs with the objective of maximizing the accuracy was studied in~\cite{Ogden2020}, where a single job is considered with a deadline constraint, and in~\cite{Nikoloska2021}, where the offloading decision is made based on an energy constraint. In contrast to the above works, we consider multiple models on the ED and provide performance guarantees for the proposed algorithm.
%Recently, \cite{Ogden2020,Nikoloska2021} studied offloading decision for inference jobs with the objective of maximizing the inference accuracy. In~\cite{Ogden2020} accuracy of a single job is considered, and in~\cite{Nikoloska2021} each job is executed on the ED and then an offloading decision is made based on a, energy constraint. the choice of multiple models on the ED was not considered.

%We assume that the the processing times of the jobs are known apriori as they can be estimated from the historical job executions. Also, we assume that communication times are known apriori and can be estimated from the data rate and the size of the data samples. This is possible for the scenarios where the ED and the ES are connected through a local area network or a private network. 
Since the true accuracy (Top-1 accuracy) provided by a model for a given data sample can only be inferred after the job is executed, for analytical tractability, we consider the average test accuracy of the model is its accuracy for any data sample. Given the processing and communication times of the jobs, we formulate the problem as an Integer Linear Program (ILP). We note that the ILP is agnostic to the actual ML models used on the ED and the ES. Different ML models on the ED may correspond to different instantiations of the same DNN with different hyperparameter values (cf.~\cite{Sandler2018}), or they may correspond to different ML algorithms such as logistic regression, support vector machines, DNN, etc.

%%%%% What are the challenges and how overcome %%%%%%%%%
Solving the formulated ILP is challenging due to the following reasons. Partitioning the set of jobs between the ED and the ES is related to scheduling jobs on parallel machines~\cite{Pinedo2008}, and assigning the jobs to the models on the ED is related to the knapsack problem~\cite{Kellerer2004}, both are known to be NP-hard. A special case of our problem, where the ED has a single model ($m=1$), is the Generalized Assignment Problem (GAP) with two machines~\cite{Ross1975}. 
%GAP is known to be APX-hard\footnote{An APX-hard problem does not have a polynomial time approximation scheme unless P$=$NP.}, 
GAP is known to be APX-hard, and the best-known approximation algorithm provides a solution that has maskespan at most $2T$~\cite{Shmoys1993}. However, the algorithms for GAP and their performance guarantees are not directly applicable to our problem due to the additional aspect that, on the ED there are multiple models to choose from. 
We propose a novel algorithm that solves a Linear Programming relaxation (LP-relaxation) of the ILP, uses a counting argument to bound the number of fractional solutions, and solves a sub-ILP problem to round the fractional solution.

Our main contributions are summarized below:
\begin{itemize}
    \item We formulate the total accuracy maximization problem subject to a constraint $T$ on the makespan as an ILP. Noting that the ILP is NP-hard,  we propose an approximation algorithm Accuracy Maximization using LP-Relaxation and Rounding (\algo) to solve it. The runtime of \algo is $O(n^3(m+1)^3)$.
    \item We prove that the total accuracy achieved by \algo is at most a small constant (less than $2$), lower than the optimum total accuracy, and its makespan is at most $2T$.  
    \item For the case of identical jobs, i.e., the data samples are identical, we propose an optimal algorithm Accuracy Maximization using Dynamic Programming (AMDP), which does a greedy packing for the ES and solves a Cardinality Constrained Knapsack Problem (CCKP) for job assignments on the ED. 
    %for which there exists an optimal algorithm based on dynamic programming that runs in pseudo-polynomial time~\cite{Kellerer2004}.
    \item As proof of concept, we perform experiments using a Raspberry Pi, equipped with MobileNets, and a server, equipped with ResNet50, that are connected over a Local Area Network (LAN). Our application is image classification for the images from ImageNet. We estimate  processing and communication times for different sizes of images, and implemented \algo and a greedy algorithm on Raspberry Pi. Our results indicate that the total test accuracy achieved by \algo is close to that of the LP-relaxed solution, and its total true accuracy is, on average, $40$\% higher than that of the greedy algorithm. 
    %\item We use simulation to compare the performance our algorithm with that of the LP-relaxed solution and also a Round-Robin Algorithm (RRA) that assigns the jobs to all the models equally. Our results indicate that the total accuracy achieved by \algo is close to the LP-relaxed solution and also is $20-30$\% higher than RRA. Furthermore, the makespan under \algo violates $T$ by at most $5$\% for the considered settings, thus indicating that for typical problem instances \algo violates the time constraint only by a small percentage.
\end{itemize}

%ML applications, such as Siri, that are deployed on smartphone obtain inference from pre-trained models in cloud. They do need internet connection. 
The rest of the paper is organized as follows. In Section~\ref{sec:related}, we present the related work. The system model is presented in Section~\ref{sec:model}. In Sections~\ref{sec:AMR2} and \ref{sec:analysis}, we present \algo and its performance bounds, respectively. In Section~\ref{sec:experiments}, we present the experimental results and finally conclude in Section~\ref{sec:conclusion}.
\section{Related Works}\label{sec:related}
In this section, we first present the related works for computation offloading problem and then discuss closely related classical job scheduling problems.
\subsection{Offloading and ML Inference Jobs}
Since the initial proposal of edge computing in~\cite{Satyanarayanan2009}, significant attention had been given to the computational offloading problem, wherein the ED needs to decide which jobs to offload, and how to offload them to an ES~\cite{Mach2017}. The objectives that were considered for optimizing the offloading decision are, 1) minimize the total execution delay of the jobs, see for example~\cite{Liu2016,Mao2016,Champati2017}, and 2) minimize the energy of the ED spent in computing the jobs, subject to a constraint on the execution delay, see for example~\cite{Chen2015,Kamoun2015,Wang2016}. However, the above works consider generic computation jobs, and the aspect of accuracy, which is relevant for the case of inference jobs, has not been considered.
%In the past few years, several works focused on developing DNN models that could be trained at the edge of in the cloud and can be finally deployed resource-constrained devices~\cite{Howard2017}

%Since the initial proposal of edge computing, there has been significant attention given to computational offloading problem. Scheduling jobs between a local device and a remote server with communication delays has been considered in~\cite{Champati2017}.

%Any connection with ML model selection?
Recently, a few works considered the problem of maximizing accuracy for  inference jobs on the ED~\cite{Wang2019,Ogden2020,Nikoloska2021}.
%In \cite{Wang2019}, the authors studied the problem of maximizing inference accuracy on an edge device subject to a deadline constraint for each frame of a video analytics application.
In \cite{Wang2019}, the authors studied the problem of maximizing the accuracy within a deadline for each frame of a video analytics application. They do not consider offloading to the edge and their solution is tailored to the DNNs that use early exits~\cite{teerapittayanon2016}. Similar problem was studied in \cite{Ogden2020}, where offloading between a mobile device and a cloud is considered. The authors account for the time-varying communication times by using model selection at the cloud and by allowing the duplication of processing the job a the mobile device. 
%In \cite{Nikoloska2021}, the authors studied the problem of scheduling inference jobs between an IoT device and a cloud, and the objective is to maximize the average accuracy subject to a maximum energy constraint in the system. The authors consider a single model on the IoT and their solution method runs the inference of each job on the IoT device and uses a confidence metric for offloading the jobs to the cloud.
A heuristic solution was proposed in \cite{Nikoloska2021} for offloading inference jobs for maximizing inference accuracy subject to a maximum energy constraint. In contrast to the above works, we consider multiple models on the ED and provide performance guarantees for \algo.

\subsection{Job Scheduling}
As noted in Section~\ref{sec:intro}, our problem is related to the knapsack problem~\cite{Kellerer2004}. To see this, note that if it is not feasible to schedule on the ES and all jobs have to be assigned to the ED, then maximizing the total accuracy is equivalent to maximizing profit, and the constraint $T$ is equivalent to the capacity of knapsack. In this case, our problem turns out to be a generalization of the CCKP~\cite{Krzyszof1989}. Another special case of our problem, where the ED has only a single model, can be formulated as a GAP \cite{Ross1975,Cattrysse1992}, with two machines. In GAP, $n$ jobs (or items) have to be assigned to $r$ machines (or knapsacks). Each job-machine pair is characterized by two parameters: processing time and cost. The objective is to minimize the total cost subject to a time constraint $T$ on the makespan.
%A feasible job assignment should respect the time constraints for finishing the jobs on each machine, and the objective is to minimize the total cost of the assignment. 
It is known that GAP is APX-hard~\cite{Chekuri2000}. 

In their seminal work \cite{Shmoys1993}, the authors proposed an algorithm for GAP that achieves minimum total cost and has makespan at most $2T$. Their method involves solving a sequence of LP feasibility problems, in order to tackle the processing times that are greater than $T$, and compute the minimum total cost using bisection search. Their algorithm can also be used for solving a related extension of GAP, where the cost of scheduling a job on a machine increases linearly with decrease in the processing time of the job. 
%The authors proposed a decision procedure, which solves an LP feasibility problem for a given total cost and outputs: 1) the problem is infeasible if no feasible schedule exists that achieves the given total cost, or 2) if a feasible schedule exists then a schedule of length $2T$ is provided. The optimal cost is then computed using a bisection search and the decision procedure for each given total cost. 
In comparison to this setting, the accuracies (equivalent to negative costs) are not linearly related to the processing times of the jobs and thus the proposed method in~\cite{Shmoys1993} is not directly applicable to the problem at hand. Our proposed algorithm \algo is different from their method in that it does not require to solve LP feasibility problems and the use of bisection search. Further, we prove the performance bounds using a different analysis technique  which is based on a counting argument for the LP-relaxation and solving a sub-problem of the ILP.

%Approximate computing, Dynamic voltage and frequency scaling of CPU clock cycles
%\input{SystemModel_PFormulation}
{\allowdisplaybreaks
\section{System Model}\label{sec:model}
Consider an ED and an ES connected over a network and the ED enlists the help of the ES for computation offloading. At time zero, $n$ inference jobs, each representing the processing requirement of a data sample on a pre-trained ML model, are available to a scheduler at the ED. Let $j$ denote the job index and $J = \{1,2,\ldots,n\}$ denote the set of job indices.

\subsection{ML Models and Accuracy}
The ED is equipped with $m$ pre-trained ML models, or simply models. Note that these may correspond to $m$ instantiations of the same DNN with different hyperparameter values resulting in different model sizes; see for example~\cite{teerapittayanon2016, Sandler2018}. Or, the models may correspond to different ML algorithms such as logistic regression, support vector machines, DNN, etc. Since the ES is a computationally powerful machine, we consider that it is equipped with a state-of-the-art model. Let $i$ denote the index of the model, and $M = \{1,2,\ldots,m,m+1\}$ denote the set of model indices, where models $1$ to $m$ are on the ED and model $m+1$ is the model on the ES. We note that our problem formulation and the solution are applicable to any family of ML models deployed on the ED and the ES.

Let $a_i \in [0,1]$ denote the average test accuracy (Top-1 accuracy) of model $i$. We note that when a job is processed on a model $i$, the resulting Top-1 accuracy, which we refer to by \textit{true accuracy} for the job, can be quite different from the average test accuracy $a_i$. However, since the true accuracies are not known apriori, for analytical tractability, we consider that the accuracy of a model $i$ for any job is $a_i$. Later, in our experimental results (cf. Section~\ref{sec:experiments}), we do present the results with the true accuracies achieved under our algorithm. Without loss of generality, we assume that $a_1 \leq a_2 \leq \ldots \leq a_m$, and also assume that the model $m+1$ is a state-of-the-art model with a higher average test accuracy than the models on the ED, i.e., $a_m \leq a_{m+1}$. 
%The above structure is not essential for the proving the performance bounds.
In the sequel, the term `accuracy' refers to the average test accuracy, unless otherwise specified. 

\subsection{Processing and Communication Times}
The processing time of job $j$ on model $i \in M\backslash\{m+1\}$ is denoted by $p_{ij}$, and on model $m+1$ it is denoted by $p'_{(m+1)j}$. Later, in Section~\ref{sec:identicaljobs}, we consider the special case of identical jobs, that is relevant in applications where the data samples are identical. In several applications, the data samples may need pre-processing before they are input to the ML model. For example, in computer vision tasks, images require pre-processing and the time required for pre-processing varies with the size of the image~\cite{Ross2019}. In our experiments with the images from the ImageNet dataset, the pre-processing stage only involves reshaping the images to input to the DNN models. Let $\tau_{ij}$ denote the pre-processing time of job $j$ on model $i$. We consider the pre-processing times are part of the processing times defined above. 

Let $c_j$ denote the communication time for offloading job $j$. It is determined by the data size of the job, i.e., the size of the data sample in bits, and the data rate of the connection between the ED and the ES. Given $p'_{(m+1)j}$ and $c_j$, the \textit{total time} to process job $j$ on the ES, denoted by $p_{(m+1)j}$, is given by $p_{(m+1)j} = c_j + p'_{(m+1)j}$. 
%\begin{align}
%p_{(m+1)j} = c_j + p'_{(m+1)j}. \label{eq:es_time}
%\end{align}
We deliberately use similar notation for the processing times $p_{ij}$ on the ED and the total times $p_{(m+1)j}$ on the ES because it simplifies the expressions in the sequel.
%Note that the difference between the processing time $p_{ij}$ on the ED and  We deliberately use $p_{(m+1)j}$ as the 
%In the sequel, we refer to $p_{(m+1)j}$ as the total time on the ES.
%Modelling the processing times to be dependent on the jobs is relevant in the applications where the data samples are not identical and may require pre-processing. For example, in computer vision tasks, the pre-processing times for images varywith the size of the image~\cite{Ross2019}. 
We consider that the communication times $c_j$ are fixed and are known apriori. %and can be estimated from the data rate and the size of the data samples.  
This is possible in the scenarios where the ED and the ES are connected in a LAN or in a private network with fixed bandwidth. In our experiments, the ED and the ES are connected via our institute's LAN, and the communication times have negligible variance. 
We also consider that the processing times of the jobs are known apriori and that they can be estimated from the historical job executions. 

%It can be advantageous in the case of Batch Scheduling. In Batch Scheduling, a machine collects a batch of tasks. Then it tries to schedule them in the system.
%We suppose that the number of jobs is always strictly greater than the number of models.

\subsection{Optimization Problem}
Given the set of jobs $J$ at time zero, the \textit{makespan} is defined as the time when the processing of the last job in $J$ is complete. The objective of the scheduler at the ED is to assign the set of jobs $J$ to the set of models $M$ such that the total accuracy, denote by $A$, is maximized and the makespan is within the time constraint $T$. Note that a schedule involves the partitioning of the set $J$ between the ED and the ES, and the constraint on the makespan implies that the completion time of all the jobs should be with in $T$ on both ED and ES. The above objective is relevant in applications where the ED periodically collects the data samples in a period $T$ and aims to finish their processing within the next period. By choosing a small period $T$, a real-time application can aim for fast ML inference at a reduced total accuracy. 

Let $x_{ij}$ denote a binary variable such that $x_{ij} = 1$, if the scheduler assigns job $j$ to model $i$, and $x_{ij} = 0$, otherwise. Note that, if $x_{(m+1)j} = 1$, then job $j$ is offloaded to the ES.  Therefore, a schedule is determined by the matrix $\mathbf{x} = [x_{ij}]$. We impose the following constraints on $\mathbf{x}$:
\begin{align}
&\sum_{i=1}^m \sum_{j=1}^n p_{ij}x_{ij} \leq T  \label{IP:eq1} \\ 
    & \sum_{j=1}^n p_{(m+1)j}x_{(m+1)j} \leq T  \label{IP:eq2}\\
    & \sum_{i=1}^{m+1}{x_{ij}=1}, \, \forall j \in J\label{IP:eq3}\\
    & x_{i,j} \in \{0 ,1\},\,\forall i \in M, \forall j \in J, \label{IP:eq4}   
\end{align}
where constraints \eqref{IP:eq1} and \eqref{IP:eq2} ensure that the total processing times on the ED and the ES, respectively, are within $T$, and thus the makespan is within $T$. Constraints in \eqref{IP:eq3} imply that each job is assigned to only one model and no job should be left unassigned, and \eqref{IP:eq4} are integer constraints.
We are interested in the following accuracy maximization problem $\P$:
\begin{align}
    &\underset{\mathbf{x}}{\text{maximize}} & & A=\sum_{i=1}^{m+1}\sum_{j=1}^n {a_ix_{ij}}\nonumber\\
    &\text{subject to} && \eqref{IP:eq1}, \eqref{IP:eq2}, \eqref{IP:eq3},  \text{ and } \eqref{IP:eq4}.\nonumber
\end{align}
\begin{comment}
\begin{align}
    &\underset{\mathbf{x}}{\text{maximize}} & & A=\sum_{i=1}^{m+1}\sum_{j=1}^n {a_ix_{ij}}\nonumber\\
    &\text{such that} && \sum_{i=1}^m \sum_{j=1}^n p_{ij}x_{ij} \leq T  \label{IP:eq1} \\ 
    &&& \sum_{j=1}^n p_{(m+1)j}x_{(m+1)j} \leq T  \label{IP:eq2}\\
    &&& \sum_{i=1}^{m+1}{x_{ij}=1}, \, \forall j \in J\label{IP:eq3}\\
    &&& x_{i,j} \in \{0 ,1\},\,\forall i \in M, \forall j \in J. \label{IP:eq4}
\end{align}
Constraints \eqref{IP:eq1} and \eqref{IP:eq2} ensure that the total processing times on the ED and the ES, respectively, are within $T$, and thus the makespan is within $T$.
Constraint \eqref{IP:eq3} implies that each job is assigned to only one model and no job should be left unassigned. 
\end{comment}
Note that $\P$ is an ILP. We will show later that a special case of $\P$ reduces to CCKP, which is NP-hard, and thus $\P$ is NP-hard. Let $A^*$ denote the optimal total accuracy for $\P$.
}
\section{Accuracy Maximization using LP-Relaxation and Rounding (AMR$^2$)}\label{sec:AMR2}
In this section, we first present the two key components of \algo: 1) the LP-relaxation of $\P$, and 2) an ILP with two jobs, which we refer by \subilp. Later, we present the details of \algo.
\subsection{LP-Relaxation and sub-ILP}
Given $\P$, we proceed with solving the LP-relaxation of $\P$, where the integer constraints in \eqref{IP:eq4} are replaced using the following non-negative constraints: \begin{align}\label{eq:relax}
x_{ij} \geq 0, \forall i \in M \text{ and } \forall j \in J.  
\end{align}
Note that, the constraints $x_{ij} \leq 1$ are not required as this is ensured by the constraints in \eqref{IP:eq3}. Let the matrix $\bar{\mathbf{x}} = [\bar{x}_{ij}]$ and $A^*_\text{LP}$ denote the schedule and the total accuracy, respectively, output by the LP-relaxation. Note that the LP-relaxed solution provides an upper bound on the total accuracy achieved by an optimal schedule, and thus we have $A^*_\text{LP} \geq A^*$.

Note that the solution to the LP-relaxation may contain $x_{ij}$ values that are fractional, and the rounding procedure is critical to proving the performance bounds. To design a rounding procedure, we first refer to a key result in~\cite{Potts1984}, where the author studied the problem of assigning $N$ jobs to $K$ parallel machines with the objective of minimizing the makespan. For the LP-relaxation of this problem, the author presented the following counting argument: there exists an optimal basic solution in which there can be at most $K-1$ \textit{fractional jobs}, i.e., the jobs that are divided between machines, and all the other jobs are fully assigned. Further,  the simplex algorithm outputs such a basic optimal solution. In our problem, there are two parallel machines, the ED and the ES, but in contrast to~\cite{Potts1984}, the ED has multiple models and the jobs assigned to the ED are processed in sequence. Taking this new aspect into account, we extend the counting argument for the problem at hand and show that solving the LP-relaxation of $\P$ results in at most two fractional jobs. This structural result is stated in the following lemma.
\begin{lemma}\label{lem:2fractionaljobs}
For the LP-relaxation of $\P$, there exists an optimal basic solution with at most two fractional jobs.
%For an LP with n equalities equal to 1 and k in-equalities constraints, a basic solution contains at most 2K non-integer basic variable.
\end{lemma}
\begin{proof}
%Let $\mathbf{z}$ denote a vector of variables and $\mathbf{b}$ denote a vector of constants. We write the $N+K$ constraints of the LP in the canonical form $\mathbf{B}\mathbf{z} = \mathbf{b}$, by introducing slack variables, such that the number of rows of $\mathbf{B}$ is equal to $N+K$. This implies that the rank of $\mathbf{B}$ is at most $N+K$
Since LP-relaxation of $\P$ has $n+2$ constraints, apart from the non-negative constraints in~\eqref{eq:relax}, one can show using LP theory that there exists an optimal basic solution with $n+2$ basic variables that may take positive values and all the non-basic variables take value zero. Under such an optimal basic solution, for the $n$ constraints in~\eqref{IP:eq3} to be satisfied, at least one positive basic variable should belong to each of those $n$ constraints. The remaining $2$ basic variables may belong to at most two equations. This implies that at least $n-2$ equations should have exactly one positive basic variable whose value should be one in order to satisfy the constraint. Therefore, there can be at most two equations with multiple basic variables whose values are in $(0,1)$, and the two jobs that correspond to these equations are the fractional jobs.
%Given the system, the conversion to canonical form ($Ax = b$) is done adding or subtracting slack variables $s_i$ to all inequalities. Now constraints are all equations. A is a matrix with $n+k$ rows and $n*(k+1)$ columns. $k<n$. All equations are linearly independent. So, we can state that the rank of matrix A is $n+k$. Following the Potts arguments in its Paper, a basic solution contains at most two equations (with equality) that contain fractional value. So we may have 2k fractional values in the basic solution.
\end{proof}

Given the basic optimal solution to the LP-relaxation, the result in Lemma~\ref{lem:2fractionaljobs} reduces the rounding procedure to assigning at most two fractional jobs. Without loss of generality, we re-index the jobs and refer to the fractional jobs by job $1$ and job $2$. We define the set $I = J \backslash \{1,2\}$ and refer to the assignment of $I$ under $\bar{\mathbf{x}}$ as the \textit{integer solution of the LP-relaxation}. We formulate the following ILP, which we refer to by sub-ILP, for computing the assignments for jobs $1$ and $2$. 
\begin{equation}\label{subILP}
\begin{aligned}
    &\underset{x_{i1},x_{i2}}{\text{maximize}} &&\sum_{i=1}^{m+1}{a_i(x_{i1}+x_{i2})} \\
    &\text{subject to} &&  \sum_{i=1}^m p_{i1}x_{i1}+p_{i2}x_{i2} \leq T \\ 
    &&& p_{(m+1)1}\,x_{(m+1)1} + p_{(m+1)2}\,x_{(m+1)2}  \leq T  \\
    &&& \sum_{i=1}^{m+1}{x_{ij}=1}, \, j \in \{1,2\}\\
    &&& x_{ij} \in \{0 ,1\}, \,\forall i \in M, \forall j \in J.
\end{aligned}
\end{equation}
Later, in Section~\ref{sec:analysis}, we will see that the \subilp in \eqref{subILP} is crucial to proving performance guarantees for \algo.
%Later, in Section~\ref{sec:analysis}, we will see that the solution of \subilp plays a key role in proving the performance guarantees. 

\subsection{$AMR^2$ Description}
%To solve the problem described in the previous section, given the processing time of each job on each machine, given the average test accuracy of each model, given a deadline T, we propose $AMR^2$. 
The main steps of \algo are summarized in Algorithm~\ref{method:AMR2}. We use $\mathbf{x}^\dagger$ and $A^\dagger$ to denote the schedule and the total accuracy, respectively, output by \algo. The schedule $\mathbf{x}^\dagger$ comprises the integer solution of the LP-relaxation and the assignment of the fractional jobs that are obtained for the cases of one fractional job and two fractions jobs in line 4 and line 7, respectively, in Algorithm \ref{method:AMR2}. 
\begin{algorithm}[ht!]
\begin{algorithmic}[1]
\caption{$AMR^2$}
\label{method:AMR2}
\STATE \textbf{Input}: $p_{ij}$, for all $i \in M$ and $j \in J$.
\STATE Solve  the LP-relaxation of $\P$.
\IF{One fractional job}
\STATE Assign job $1$ to model $\argmax_{i\in M}\{a_i : p_{i1} \leq T\}$.
\ENDIF
\IF{Two fractional jobs}
\STATE Solve the \subilp in~\eqref{subILP} using Algorithm \ref{algorithm:sub_ilp}.
\ENDIF
\STATE \textbf{Output}: Assignment matrix $\mathbf{x}^\dagger$ and total accuracy $A^\dagger$ 
\end{algorithmic}
\end{algorithm}

The LP-relaxation is solvable in polynomial time. To solve the \subilp we consider different cases based on the processing times of jobs 1 and 2, and greedily pack them to maximize accuracy subject to the constraint $T$. 
%Based on the case, we propose scheduling criterion, we calculate the completion time of each machine, we compare the total accuracy reached with \algo and the $A^*_\text{LP}$. 
The steps for solving the \subilp are presented in Algorithm~\ref{algorithm:sub_ilp}. There are two main cases. For the case where processing time of at least one of the jobs is at most $T$ on the ES, we assign at least one job to the ES. This case is presented in line $2$ of Algorithm~\ref{algorithm:sub_ilp} and is satisfied for problem instances where $p_{ij} \leq T$, for all $i \in M$ and $j \in J$. For the case where processing times of both the jobs are greater than $T$, we schedule both of them on the ED. This case is presented in line 12 of Algorithm~\ref{algorithm:sub_ilp}. In line 13, we use enumeration for finding the models $i'$ and $i''$.
%We use existing solvers for solving the LP-relaxation, while the steps for solving \subilp are presented in Algorithm~\ref{algorithm:sub_ilp}. 
%Later, we will show that Algorithm~\ref{algorithm:sub_ilp} computes the optimal solution for \subilp. 
%Note that if the number of fractional jobs is 1, the best scheduling policy is assigning the job $1$ to model $\argmax_{i}\{a_i : p_{i1} \leq T\}$. 
\begin{algorithm}[ht!]
\caption{Algorithm for solving \subilp}
\label{algorithm:sub_ilp}
\begin{algorithmic}[1]
\STATE \textbf{Input}: $p_{i1}$ and $p_{i2}$, for all $i \in M$
\IF{$p_{(m+1)1} \leq$ T \OR $p_{(m+1)2} \leq T$}
%\STATE Assign job with $p_{(m+1)j}\leq T$ to model $m+1$.
    \IF{$p_{(m+1)1}+p_{(m+1)2}\leq T$}
    \STATE Assign both jobs to model $m+1$.
    \ELSE 
    \IF{$\max\{a_i:p_{i1}\leq T\} \geq \max\{a_i:p_{i2}\leq T\}$}
    \STATE Assign job $1$ to model $\argmax_{i}\{a_i:p_{i1}\leq T\}$ and job $2$ to the ES.
    \ELSE
    \STATE Assign job $2$ to model $\argmax_{i}\{a_i:p_{i2}\leq T\}$ and job $1$ to the ES.
    \ENDIF
    %pick one model on ED \\ i=$\max(\arg_i\max\{a_i:p_{i1}\leq T\},\arg_i\max\{a_i:p_{i2}\leq T\} )$
    %\STATE assign job J such that  $p_{ij}\leq T$
    %\STATE assign other job to model m+1
    \ENDIF
\ELSIF{$p_{(m+1)1} > T$ \AND $p_{(m+1)2} > T$}
\STATE Assign job $1$ to $i'$ and job $2$ to $i''$, where models $i'$ and $i''$ are on the ES such that, $p_{i'1}+p_{i''2}\leq T$ and $a_{i'} + a_{i''}$ is the maximum.
%Enumerate the possible combination where $p_{i'1}+p_{i''2}\leq T$ where $i,i'' \in \{1,..,m\}$ and calculate the sum of the accuracy of the combinations.
%\STATE Pick the combination with higher accuracy.
\ENDIF
\STATE \textbf{Output}: Assignment for jobs 1 and 2. 
\end{algorithmic}
\end{algorithm}

\subsubsection*{Computational complexity}
The computational complexity of solving an LP with $l$ variables is $O(l^3)$ (cf.\cite{Jan2020}). In the LP-relaxation, the number of variables are $n(m+1)$, and thus its runtime is $O(n^3(m+1)^3)$. Solving \subilp using Algorithm~\ref{algorithm:sub_ilp} requires $m^2$ iterations due to the step in line $14$, where we enumerate $m$ choices for each job. Therefore, the runtime of \algo is $O(n^3(m+1)^3)$.
%As matter of fact the LP-relaxation problems is solvable with a complexity time that is equal to the number of variables per cube. Since the number of variables here is $n(m+1)$ we can state that the lp-relaxation takes $(n(m+1))^3$ time. The algorithm for solving \subilp takes in the first case linear time. In the second case ($p_{(m+1)1} >= T$ and $p_{(m+1)2} >= T$) the algorithm requires to enumerate the combinations. This takes $O(m^2)$ times. Since $n>m$, we can state that the computational time of the algorithm is $O(n(m+1))^3$.

% \textbf{Output}: the integer assignment of fractional jobs, the total average accuracy that the two jobs reach 

\section{Analysis of \algo}\label{sec:analysis}
 In this section, we analyse \algo and present a $2T$ bound for its makespan and its total accuracy is at most $2(a_{m+1}-a_1)$ lower than the optimal accuracy. %Toward this end, we first present some auxiliary results.

%We prove that the algorithm \ref{algorithm:sub_ilp} maximizes accuracy and gives solutions within time T. 
We designed Algorithm~\ref{algorithm:sub_ilp} for computing an optimal solution, with makespan at most $T$, for the \subilp by considering all possible cases. From its description one can easily verify that it is indeed optimal for \subilp. We state this in the following lemma without proof.
\begin{lemma}\label{lem:Algo2Opt}
Algorithm \ref{algorithm:sub_ilp} is an optimal algorithm for the \subilp.
\end{lemma}

In the following theorem, we state the bound for the makespan under \algo.
\begin{theorem}\label{theorem:makespan2T}
If $\P$ is feasible, then the makespan of the system under \algo is at most 2T. 
\end{theorem}
\begin{proof}
Given that $\P$ is feasible, the LP-relaxation of $\P$ is also feasible. Since the makespan under the LP-relaxation solution is at most $T$, the makespan of its integer solution cannot exceed $T$. Again, the feasibility fo $\P$ implies that the \subilp is feasible as it is a sub problem of $\P$ involving only two jobs from $J$. Using this and Lemma~\ref{lem:Algo2Opt} we infer that Algorithm~\ref{algorithm:sub_ilp} always finds an assignment for jobs $1$ and $2$ such that the makespan for these two jobs is at most T. Now, $\mathbf{x}^\dagger$ comprises the integer solution of the LP-relaxation, with makespan at most $T$, and the assignment of the fractional jobs in both the cases of one fractional job and two fractional jobs, which also have makespans of at most $T$. Therefore, \algo has a makespan at most 2T.
\end{proof}

\begin{theorem}\label{theorem:accuracy_constant}
The total accuracy achieved by an optimal schedule is at most $2(a_{m+1}-a_1)$ higher than the total accuracy achieved by \algo, i.e., $A^* \leq A^\dagger + 2(a_{m+1}-a_1)$. 
\end{theorem}
\begin{proof}
We prove that $A^*_\text{LP} \leq A^\dagger + 2(a_{m+1}-a_1)$ in the worst-case, and the result follows from the relation $A^*_\text{LP} \geq A^*$. If there are no fractional jobs in the LP-relaxed solution, then the result is true since in this case $A^*_\text{LP} = A^\dagger$. If there is one fractional job, which we refer by job $1$, then the difference between $A^*_\text{LP}$ and $A^\dagger$ is caused by reassigning job $1$ by \algo (cf. line 4 Algorithm~\ref{method:AMR2}). Since job $1$ can contribute at most $a_{m+1}$ to $A^*_\text{LP}$ and its reassignment by \algo results in at least a contribution of $a_1$ to $A^\dagger$, we obtain $A^*_\text{LP} - A^\dagger \leq a_{m+1} - a_1$. 

When there are two fractional jobs, for some models $i$, $\hat{i}$, $k$ and $\hat{k}$, it should be true that $\bar{x}_{i1} + \bar{x}_{\hat{i}1} = 1$ and $\bar{x}_{k2} + \bar{x}_{\hat{k}2} = 1$, meaning that job $1$ is divided between models $i$ and $\hat{i}$ and job $2$ is divided between models $k$ and $\hat{k}$ under the LP-relaxed solution $\bar{\mathbf{x}}$. Let $A^*_{\text{LP,I}}$ denote the total accuracy for the set $I = J\backslash\{1,2\}$ under $\bar{\mathbf{x}}$. 
%i.e., the total accuracy of the integer solution of the LP-relaxation.
We have,
\begin{align}
A^*_\text{LP} &= A^*_\text{LP,I} + a_{i}\Bar{x}_{i1} + a_{\hat{i}}\Bar{x}_{\hat{i}1} + a_{k}\Bar{x}_{k2}+a_{\hat{k}}\Bar{x}_{\hat{k}2} \label{lp_accuracy1}\\
&\leq A^*_\text{LP,I} + 2a_{m+1}. \label{lp_accuracy2}
\end{align}
In the last step above, we used $a_{m+1} \geq a_i$, for all $i \in M$. 
We now consider different cases based on the job assignments output by the Algorithm~\ref{algorithm:sub_ilp}.

\textbf{Case 1:} Both jobs are offloaded to model $m+1$. Since \algo comprises the integer solution solution of the LP-relaxation, i.e., the schedule of $I$ under $\bar{\mathbf{x}}$, and the assignment from the Algorithm~\ref{algorithm:sub_ilp}, we have $A^\dagger = A^*_\text{LP,I} + 2a_{m+1}$. Using~\eqref{lp_accuracy2}, we obtain $A^\dagger \geq A^*_\text{LP}$.
%\begin{align*}
%    A^\dagger = A^*_\text{LP,I} + 2a_{m+1} \geq A^*_\text{LP}.
%\end{align*}
%The last step above follows from \eqref{lp_accuracy2}. 
Note that $A^\dagger$ could exceed $A^*_\text{LP}$ because the makespan of $\mathbf{x}^\dagger$ can exceed $T$ in contrast to $\bar{\mathbf{x}}$.

\textbf{Case 2:} One job is offloaded on model $m+1$ and one job is assigned to model $\hat{i}\in J \backslash \{m+1\}$ on the ED. This implies, the offloaded to the ES contributes $a_{m+1}$ and the other job contributes at least $a_1$ to to the total accuracy $A^\dagger$. Therefore, we have,
\begin{align*}
    A^\dagger \geq A^*_\text{LP, I} + a_{m+1} + a_1 \geq  A^*_\text{LP} -(a_{m+1} - a_1).
\end{align*}
In the second step above we used \eqref{lp_accuracy2}. 

\textbf{Case 3:} Both jobs are assigned to the ED. This case occurs when $p_{(m+1)1} > T$ and $p_{(m+1)2} > T$, and we have
\begin{align}\label{eq:i'i''}
    A^\dagger \geq A^*_\text{LP,I} + 2a_1.
\end{align}
From~\eqref{eq:i'i''} and \eqref{lp_accuracy2}, we obtain $A^*_\text{LP}-A^\dagger \leq 2(a_{m+1}-a_1)$.

We obtain the result by taking the worst case bounds in all the three cases.
\end{proof}
Since $2(a_{m+1}-a_1) \leq 2$, the difference between $A^*$ and $A^\dagger$ becomes negligible as the number of jobs $n$ increases. It is worth noting that the worst-case bound $2(a_{m+1}-a_1)$ results from \textbf{Case 3} in the proof of Theorem~\ref{theorem:accuracy_constant}. While this case, where the processing times of jobs $1$ and $2$ are greater than $T$, can happen when $T$ is small, typical problem instances have processing times less than $T$ on the ES. Note that a $T$ value that is smaller than the processing time of a job either becomes infeasible or results in very low optimal total accuracy. We also assert this in our experimental results. In the following corollary we present the worst-case bound for a case that is true for typical problem instances.
\begin{corollary}\label{cor:accuracyBound}
If the processing times of all jobs on the ES are at most $T$, then $A^* \leq A^\dagger + a_{m+1} - a_1$.
\end{corollary}
\begin{proof}
Since processing times of all jobs on the ES are at most $T$, \textbf{Case 3} in the proof of Theorem~\ref{theorem:accuracy_constant} cannot happen. In all other cases, the worst-case bound we have is $A^*_\text{LP} \leq A^\dagger + a_{m+1} - a_1$, and the result is true since $A^* \leq A^*_\text{LP}$.
\end{proof}

\textbf{\textit{Remark:}} The schedule $\mathbf{x}^\dagger$ given by \algo may result in a makespan greater $T$. As noted before, a special case of our problem is GAP for which the best known approximation algorithm, proposed in~\cite{Shmoys1993}, has the makespan bound $2T$ and produces a schedule that may exceed $T$. In our experimental results, we show that the percentage of violation on an average is at most $40$\% for the considered image classification application. 
%For application scenarios where the the makespan constraint violation cannot be tolerated, on may choose to schedule schedule the two fractional jobs from the LP-relaxed solution 
%The authors in~\cite{Chekuri2000} modified this approximation algorithm to produce a schedule which respect the time constraint, but they achieve this by dropping some jobs.

%With theorem (\ref{theorem:makespan2T}) we proved that \algo incurs a makespan at most 2T. With theorem (\ref{theorem:accuracy_constant}) we proved that the difference between total accuracy achieved under the optimal policy and \algo is within a  $a_{m+1}$.

%Note that, as $n$ increase, and also the constraint $T$ increased accordingly, then this constant difference becomes negligible. Furthermore, if the processing time are for all the jobs less than the constraint T, it is highlighted in the Simulation Results section that \algo violates the constraint T at most of 4.5\%.

{\allowdisplaybreaks
\section{Identical Jobs}
In this section, we consider the problem $\PI$, a special case of $\P$ where the jobs are identical, i.e., $p_{ij} = p_{i}$, for all models $i \in M$. We present Accuracy Maximization using Dynamic Programming (AMDP) for $\PI$. The formulation for $\PI$ is given below.
\begin{align}
& \underset{\mathbf{x}}{\text{maximize}} & & \sum_{i=1}^{m+1}a_i \sum_{j=1}^{n}x_{ij} \nonumber \\
& \text{subject to} & & \sum_{i = 1}^m p_i \sum_{j=1}^n x_{ij} \leq T 	\label{constraint01}\\
& & &p_{m+1} \sum_{j \in 1}^n x_{(m+1)j} \leq T \label{constraint02}\\
& & & \sum_{i = 1}^{m+1} x_{ij} = 1, \quad \forall j \in J \label{constraint03}\\
& & & x_{ij} \in \{0,1\}, \quad \forall i \in M, \forall j \in J. \label{constraint04}
\end{align}
Next, we exploit the structure of $\PI$ and reduce it to solving a Cardinality Constrained Knapsack Problem (CCKP).  

\subsection{CCKP}\label{sec:identicaljobs}
Our first observation is that the number of jobs assigned to the ES under an optimal schedule is given by $n_c = \floor{\frac{T}{p_{m+1}}}$. To see this, assigning number of jobs less than $n_c$ can only reduce the accuracy as the ES provides highest accuracy, and no more than $n_c$ can be assigned due to constraint~\eqref{constraint02}. We present this observation in the following lemma.
\begin{lemma}\label{lem:opt_cloud_sched}
Under an optimal schedule, the number of jobs assigned to the ES is given by $ n_c =  \floor{\frac{T}{p_{m+1}}}$. 
\end{lemma}
We define $n_l = n - n_c$. Since the jobs are identical, without loss of generality, we assign the last $n_c$ jobs to the ES. 
%Therefore, we have
%\[
%x^*_{m+1\,j} = 
%\begin{cases} 
%1 & j \geq n_l + 1  \\
%0 & \text{ otherwise.} 
%\end{cases}
%\]
We are now only required to compute the optimal assignment for jobs $j \in \{1,\ldots,n_l\}$ to the models $1$ to $m$ on the edge device. Thus, given Lemma~\ref{lem:opt_cloud_sched}, solving $\PI$ is reduced to solving the following problem $\PI'$:
\begin{align}
& \underset{\mathbf{x}}{\text{maximize}} & & \sum_{i=1}^{m}a_i \sum_{j=1}^{n_l}x_{ij} \nonumber \\
& \text{subject to} & & \sum_{i = 1}^m p_i \sum_{j=1}^{n_l} x_{ij} \leq T, 	\label{constraint11}\\
& & & \sum_{i = 1}^{m} x_{ij} = 1, \quad \forall j \in \{1,\ldots,n_l\}   \label{constraint12}\\
& & & x_{ij} \!\in \!\{0,1\}, \, \forall i \in M \backslash \{m+1\}, \forall j \in \{1,\ldots,n_l\}. \nonumber
\end{align}
We do a variable change to formulate the CCKP. Let $r$ denote an index taking values from $\{1,\ldots,mn_l\}$. We define new variables $z_r$, accuracies $\bar{a}_r$, and processing times $\bar{p}_r$ as follows: for $i \in M \backslash \{m+1\}$ and $j \in \{1,\ldots,n_l\}$,
\begin{align*}
& z_r = \{x_{ij}: r = j + (i-1)n_l\},\\
&\bar{a}_r = a_i, \, (i-1)n_l \leq r < in_l, \\
&\bar{p}_r = p_i, \, (i-1)n_l \leq r < in_l.
\end{align*}
The CCKP using $\{z_r:1\leq r \leq m n_l\}$ as the decision variables is stated below.
\begin{align}
& \underset{\{z_r\}}{\text{maximize}} & & \sum_{i=1}^{mn_l}\bar{a}_r z_r \nonumber \\
& \text{subject to} & & \sum_{r = 1}^{mn_l} \bar{p}_r z_r \leq T, 	\label{constraint21}\\
& & & \sum_{r = 1}^{mn_l} z_{r} = n_l,   \label{constraint22}\\
& & & z_{r} \!\in \!\{0,1\}, \, \forall r \in \{1,\ldots,mn_l\}. \label{constraint23}
\end{align}
Let $\{z^*_r:1\leq r \leq mn_l\}$ denote an optimal solution for CCKP. The CCKP can be interpreted as follows. We have $mn_l$ items, where each item represents a model and there are $n_l$ copies of the same model. Since the jobs are identical, the problem reduces to the number of times a model is selected, equivalent to the number of jobs assigned to it, such that all jobs are assigned. In the following lemma we state that solving CCKP results in an optimal solution for $\PI'$.
\begin{lemma}\label{lem:equivalence}
The solution $x^*_{ij} = \{z^*_r: r = j + (i-1)n_l\}$
%\begin{align*}
%    x^*_{ij} = \{z^*_r: r = j + (i-1)n_l\}
%\end{align*}
is an optimal solution for $\PI'$. 
\end{lemma}
\begin{proof}
By construction, $\PI'$ and CCKP have one-to-one mapping between the decision variables, have equivalent objective functions and constraints in~\eqref{constraint11} and~\eqref{constraint21}. They only differ in the constraints~\eqref{constraint12} and~\eqref{constraint22}. We note that \eqref{constraint22} is equivalent to 
\begin{align}\label{sum_constraint}
    \sum_{j=1}^{n_l}\sum_{i = 1}^{m} x_{ij} = n_l.
\end{align}
%which is a relaxation of the constraint in~\eqref{constraint12}. 

Let $\PI^\ddagger$ denote the problem $\PI'$ with the constraint~\eqref{constraint12} replaced by~\eqref{sum_constraint}. From the above observations, $\PI^\ddagger$ is equivalent to CCKP, and thus it is sufficient to show that an optimal solution $\{x^\ddagger_{ij}\}$ for $\PI^\ddagger$ is optimal for $\PI'$. Since~\eqref{sum_constraint} is a relaxation of the constraint in~\eqref{constraint12}, the optimal objective value of $\PI^\ddagger$ should be at least the optimal objective value of $\PI'$. On the other hand, given $\{x^\ddagger_{ij}\}$, consider the assignment where for each model $i$, we assign $\sum_{j=1}^{n_l}x^\ddagger_{ij}$ jobs to it. Given that the jobs are identical, and from~\eqref{sum_constraint}, all the $n_l$ jobs will be assigned exactly once to some model. Thus, this assignment is feasible for $\PI'$ and objective value under this assignment will be equal to the optimal objective value of $\PI^\ddagger$. Thus, $\{x^\ddagger_{ij}\}$ is also an optimal solution for $\PI'$. \end{proof}

In Algorithm~\ref{method:PI} we present AMDP for solving $\PI$. The optimality of AMDP is a direct consequence of Lemmas~\ref{lem:opt_cloud_sched} and~\ref{lem:equivalence} and is stated in the following theorem.
\begin{theorem}
AMDP is an optimal algorithm for $\PI$.
\end{theorem}
\begin{algorithm}
\begin{algorithmic}[1]
\caption{AMDP}
\label{method:PI}
\STATE $n_l = n - \floor{\frac{T}{p_{m+1}}}$
\STATE Assign the jobs $j \in \{n_l+1,\ldots,n\}$ to the ES %x^*_{m+1\,j} = 1, \text{ for } j \in \{n_l+1,\ldots,n\}
\STATE Solve the CCKP for $\{z^*_r\}$ using the DP algorithm
\STATE Assignment for remaining jobs: $x^*_{ij}\! = \!\{z^*_r: r = j\! +\! (i-1)n_l\}$ for all $i \in M \backslash \{m+1\}$, and $j \in \{1,\ldots,n_l\}$.
%\FOR{$i = 1$ to $m$}
%\STATE From the remaining jobs, assign $\sum_{r=}z^*_r$ jobs to model $i$.
%\ENDFOR
\end{algorithmic}
\end{algorithm}

\subsection{The DP Algorithm}
The main step in AMDP is to solve the CCKP for which one can leverage existing branch-and-bound or Dynamic Programming (DP) algorithms~\cite{Kellerer2004}. We use the DP algorithm since it has pseudo-polynomial runtime for computing the optimal solution. 
The summarize the main steps of the algorithm. Let $s$, $k$, and $\tau$ denote positive integers. We define $y_s(\tau,k)$ as the maximum accuracy that can be achieved by selecting items from the set $\{1,\ldots,s\}$, where $s \leq mn_l$, given a time constraint $\tau$ ($\leq T$) and the number of items to be selected are $k$ ($\leq n_l$).
\begin{align}
    y_s(\tau,k)\! = \!\max\left\{\!\sum_{r=1}^{s}\bar{a}_r z_r\Big{|}\sum_{r = 1}^{s}\! \bar{p}_r z_r \leq \tau,\!\! \sum_{r = 1}^{j} z_r = k, z_r\!\! \in\!\! \{0,1\}\!\!\right\}
\end{align}
The DP iterations are given below:
\[
y_s(\tau,k)\! =\! 
\begin{cases} 
y_{s-1}(\tau,k) & \text{if } \bar{p}_s \geq \tau  \\
\max\{y_s(\tau\!-\!\bar{p}_s,k\!-\!1)\!+\!\bar{a}_s,y_{s-1}(\tau,k)\} & \text{otherwise.} 
\end{cases}
\]    
We compute the solution for $y_s(T,n_l)$, where $s = mn_l$. 

The computational complexity of the DP algorithm is $O(mnT)$ and AMDP has the same computational complexity. 

\textit{\textbf{Remark:}} We note that AMDP can be easily extended to a slightly more general setting where the processing times of the jobs only depend on the models, but they may have different communication times. This implies that $p_{ij} = p_{i}$ for all the models on the ED, and $p'_{(m+1)j}=p'_{(m+1)}$ for all $j$, but $c_j$ values could be different. This setting is applicable in scenarios where the data samples are heterogeneous, but their processing times on the models do not vary much. In this case, we order the list of jobs in the increasing order of their communication times and offload the jobs from the start of the list to the ES until constraint $T$ is met. Since the jobs have same processing times on the models of the ED, it is easy to argue that this assignment optimal for the ES. For the remaining jobs we solve the CCKP.
}
\section{Experimental Results}\label{sec:experiments}
 %The data samples are images from the ImageNet test set\cite{Imagenet2009}.
 In this section, we first present the experimental setup. We then present the implementation details for estimating the processing and communication times. As explained in Section~\ref{sec:related}, the aspect of multiple models on the ED has not been considered in computation offloading literature and  there are no existing algorithms that are applicable for the problem at hand for a performance comparison. Therefore, we present the performance comparison between \algo and a baseline Greedy Round Robin Algorithm (\blalgo). Given the list of jobs, \blalgo offloads them from the start of the list to the ES until the constraint $T$ is met. The remaining jobs are assigned in a round robin fashion to the models on the ED until the constraint $T$ is met. Any further remaining jobs are assigned to model $1$. Note that \blalgo solution may violate the time constraint $T$ and its runtime is $O(n)$.

\subsection{Experimental Setup}
\begin{comment}
Their configurations of Raspberry Pi and the server are presented in Table \ref{tab:devices}. 
\begin{table}[ht!]
    \centering
    \begin{tabular}{||c| c |c |c ||} 
    \hline
    Device & Cores & Frequency & RAM \\ 
    \hline\hline
    Raspberry Pi  4 & 4 & 1.5 Ghz & 4GB  \\ 
    \hline
    Server & 512 & 1.4 Ghz & 504 GB \\
    \hline
    \end{tabular}
    \caption{In this table we present Hardware specification of devices used in the experiments.}
    \label{tab:devices}
\end{table}.
\end{comment}
%The Raspberry Pi 4 presents an architecture aarch64 and as operative system Raspbian 10. 
%The LP-relaxation problem of P is solved using PuLP library. 
%The Operative System used is Debian 11. Requests are sent from Raspberry Pi to the sever via HTTP protocol. 
Our experimental setup comprises a Raspberry Pi device (the ED) and a local server (the ES) that are connected and located in the same LAN.
Raspberry Pi has $4$ cores, $1.5$ GHz CPU frequency, and $4$ GB RAM, with the operating system Raspbian 10, while the server has $512$ cores, $1.4$ GHz CPU frequency, and $504$ GB RAM, with the operating system Debian 11. All the functions on Raspberry Pi and on the server are implemented using Python 3. To offload images from Raspberry Pi to the ES we used HTTP protocol, and implemented HTTP Client and Server using Requests and Flask, respectively. 

%We consider image classification for classifying images from the ImageNet dataset. 
The data samples are images from the ImageNet dataset for which we use DNN models for inference. On Raspberry Pi we import, from the TensorFlow Lite library, two pre-trained MobileNets corresponding to two values $0.25$ and $0.75$ for the hyperparameter $\alpha$, which is a width multiplier for the DNN~\cite{Howard2017}. Both the models are quantized and require input images of dimensions $128\times 128$. On the ES, we import a pre-trained ResNet50 model~\cite{He2016} from the Tensorflow library. The ResNet50 model requires input images of dimensions $224 \times 224$. Images of different dimensions need to be reshaped to the respective dimensions on the ES and the ED. The average test accuracies for the three models are presented in Table~\ref{tab:avg_accuracies}.
%We used 3 DNN models: two on the Raspberry Pi  and one on the ES.All models are pre-trained using ImageNet\cite{Imagenet2009}.These models are available from TensorFlow Library. DNN models on the edge device are two MobileNets\cite{Howard2017} DNN. MobileNets introduces an  hyperparameter $\alpha$. $\alpha$, or width multiplier, is a multiplying factor that uniformly thin each layer of the DNN. The model on the ED are quantized. Quantization describes the process of reducing the precision of the weights. Thus, it is reduced the size of the DNN model.
\begin{table}[ht!]
    \centering
    \begin{tabular}{|c | c|} 
        \hline
        \textbf{Model} & \textbf{Average test Accuracy} \\ 
        \hline\hline
        MobileNet $\alpha=0.25$ (model $1$) & $0.395$ \\ 
        \hline
        MobileNet $\alpha=0.75$ (model $2$) & $0.559$\\
        \hline
        ResNet50 (model $3$) & $0.771$ \\
        \hline
    \end{tabular}
    \caption{Test accuracies of the considered DNN models~\cite{TFLite}.}
    \label{tab:avg_accuracies}
\end{table}
%All the models on the Raspberry Pi needs images of dimensions $128\times 128$. The DNN on ES requires images of 224x224 dimensions. Images of different dimensions need to be reshaped to the respective dimensions on the the ES and the ED.  We decided to classify images of dimension: $128\times 128$, $512\times 512$, ,$1024\times 1024$. 

We implemented both \algo and \blalgo on Raspberry Pi in Python 3. \algo takes up to $50$ ms for computing a schedule for $40$ jobs. 
%This runtime can be significant when the time constraint $T$ is small, for example $T = 0.5$ sec. 
The runtime of \algo is dominated by the runtime for solving LP-relaxation. 
%In our current implementation, we call a library function in Python~\cite{CBC} to solve the LP-relaxation. 
In future, we plan to reduce this runtime by implementing \algo in C. We, however, implemented AMDP in C and observed that it has a runtime less than $1$ ms on Raspberry Pi for computing an optimal schedule for $300$ jobs. This shows the advantage of using AMDP over \algo for the case of identical jobs.

\textbf{\textit{Remark:}} The testbed that we implemented uses a single thread on the Raspberry Pi. To further improve the performance of this system one may use two different threads, one for handling the offloading of the tasks to different models and another for retrieving the responses and rendering the results. This implementation improves the quality of experience of the user as the jobs will be rendered as soon as they are processed.

\subsection{Estimation of Processing and Communication Times}
In our experiments, we consider images of dimensions $128\times 128$, $512\times 512$, and $1024\times 1024$, for which we estimate the processing and communication times using the following procedure. On Raspberry Pi, we run $30$ samples of same image dimensions and use the median processing times as our estimate. Note that median is an unbiased estimate, and unlike the mean, it is not affected by cold start. We note that the estimates for the processing times include the reshape times.
%We needed to characterize the processing time of the inference. The total time needed to run an inference on Raspberry Pi is $p_{ij}$. In $p_{ij}$ we consider the time needed for reshaping an image to the input size of the DNN, which in this case is $224\times 224$,, plus the processing time needed for having a classification output from a model $i$ given an image $j$. The reshaping time on a model on the edge device in this particular case where the input size is the same for all the model depends only from the input size of the image.

In order to estimate the total time on the ES, we use the HTTP client/server connection to send $30$ images of same image dimensions from Raspberry Pi to the server. For each image we measure the time till the reception of an inference for the image from the ES, and finally use the median. At the server we also measure the reshape time and the processing time, and the estimate for the communication time is obtained by subtracting the reshape time and the processing time from the total time. Since Raspberry Pi and the dedicated local server are in the same LAN, the observed communication times are almost constant with negligible variance. This is also true for the observed processing times, and we will later verify this when implementing the schedules using these estimates. 

The estimates for the processing times are presented in Table \ref{tab:processing_times}. Observe that the processing times increase with the model size. On Raspberry Pi, the variance in the processing times on a model is small. In contrast, the total times on the ES vary with the dimensions of the image and are an order of magnitude higher than the processing times on Raspberry Pi. In Figure~\ref{fig:timeErdos}, we present the communication, reshape, and processing times on the ES. It is worth noting that, as the dimensions of the image increases, both communication time and the reshape times increase. Thus, it is more advantageous to offload images with smaller dimensions. 
%The bandwidth of the communication of the system has been studied using iPerf3 with 30 experiments of 1 minutes each. Note that the measured bandwidth varies based on the chosen Tranport layer. We use HTTP to send images from ED to ES. Thus, we measured TCP bandwidth. The average bandwidth is 987 Mbit/s. However, using Requests and Flask we achieved an average bandwidth of 144 Mbit/s. The most important think that we realized is that the bandwidth is constant within the experiments. Given as input a JPEG image of $1024\times1024$ to the network, we experienced an average communication time of 0.0007 sec with a variance of $2.88\times10^-5$.
%This is because we are using a point to point connection.

%The total time $p_{(m+1)j}$ needed to run the inference on the ES is (\ref{eq:es_time}). These times, given an average bandwidth,$c_j$ varies based on the size of an image. In our system we decided to reshape the images on the Edge Server. Times have been recorded using a primitive present in Requests. This primitive calculates the time between the creation and the closing of a TCP socket on the Client Side. Since reshaping and processing is done on the Edge Server, all the factors explained in equation (\ref{eq:es_time}) are considered. The processing time for each image size on a model is obtained by median of the runtimes from 30 experiments.
%Table~\ref{tab:processing_times}.
\begin{table}[ht!]
    \centering
        \begin{tabular}{|c|c|c|c|} 
            \hline
             \textbf{Model/Image Dimension} & $128\times128$ & $512\times512$& $1024\times1024$\\ 
            \hline\hline
            MobileNet $\alpha=0.25$ & $0.01$ & $0.011$ & $0.011$ \\ 
            \hline
            MobileNet $\alpha=0.75$ & $0.04$ & $0.04$ & $0.043$\\
            \hline
            ResNet50 & $0.28$ & $0.32$ & $0.38$\\
            \hline
        \end{tabular}
        \caption{Estimated processing times (in seconds), for MobileNets on Raspberry Pi and ResNet50 on the server.}
    \label{tab:processing_times}
\end{table}

\begin{figure}[ht!]
    \centering
    \vspace{-0.7cm}
    \includegraphics[width=2.8in]{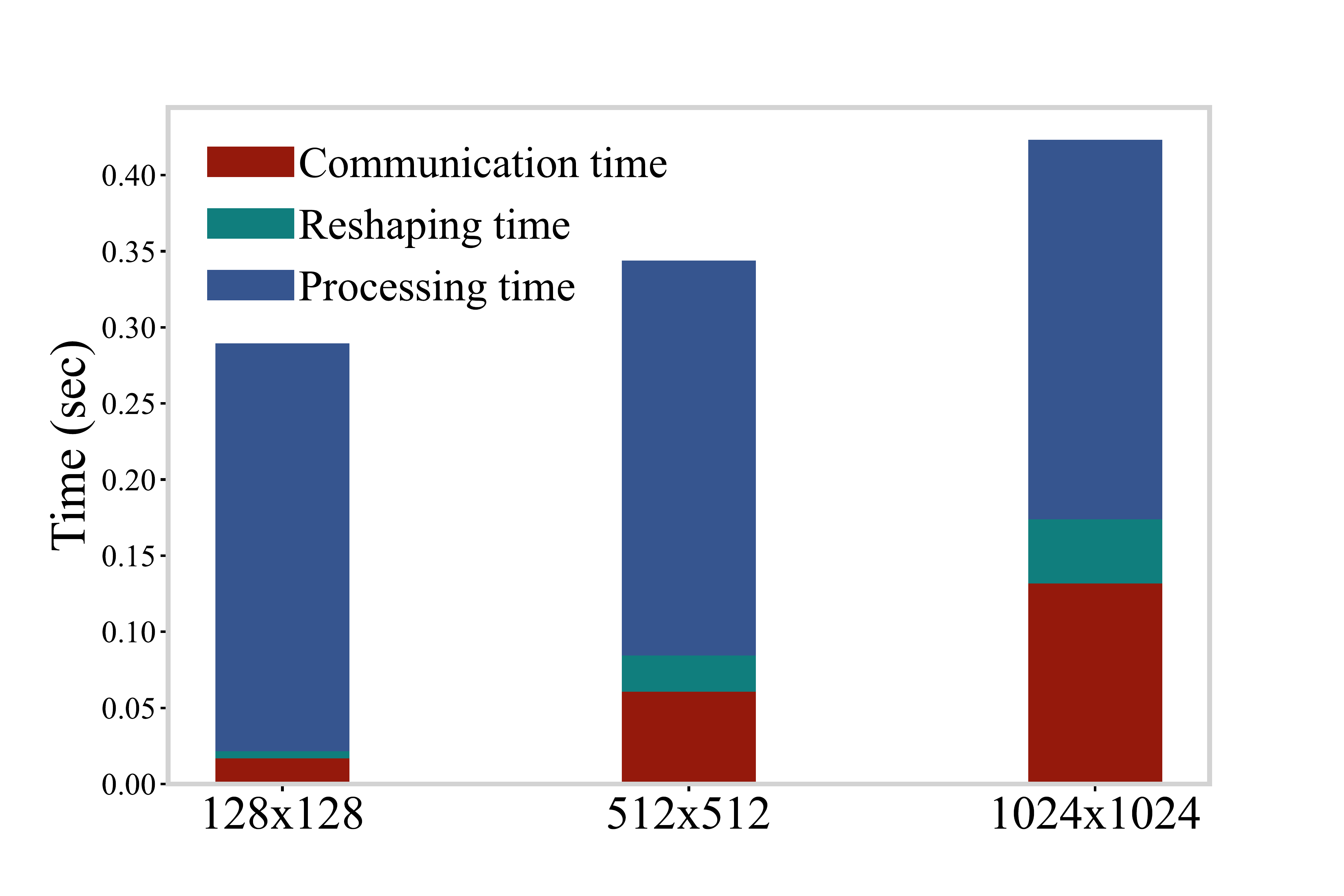}
    \caption{Estimated total time for inference on the ES.}
    \label{fig:timeErdos}
\end{figure}
%In Figure \ref{fig:timeErdos} we present $p_{(m+1)j}$ time in seconds when images have dimensions $128\times128$, $512\times512$, $1024\times1024$. The processing time after the reshaping is identical because it is the median time of the model. The variable that impacts most $p_{(m+1),j}$ is the connection time.
%This is because we have a connection with a fixed bandwidth and the size of the images varies. Thus, Connection Time increases as the size of the image increases.

%The LP solver used is Coin or Base and Cut(CBC)\cite{CBC}. The CBC solver, given 40 jobs and a time constraint T=2 solves the LP problem in average 0.05 seconds. This performances are reasonable for our problem and our time constraint.

\subsection{Performance of \algo}
\begin{figure}[ht!]
\centering
\vspace{-2.5cm}
\hspace*{-1cm}                                                           
\includegraphics[width=4.0in]{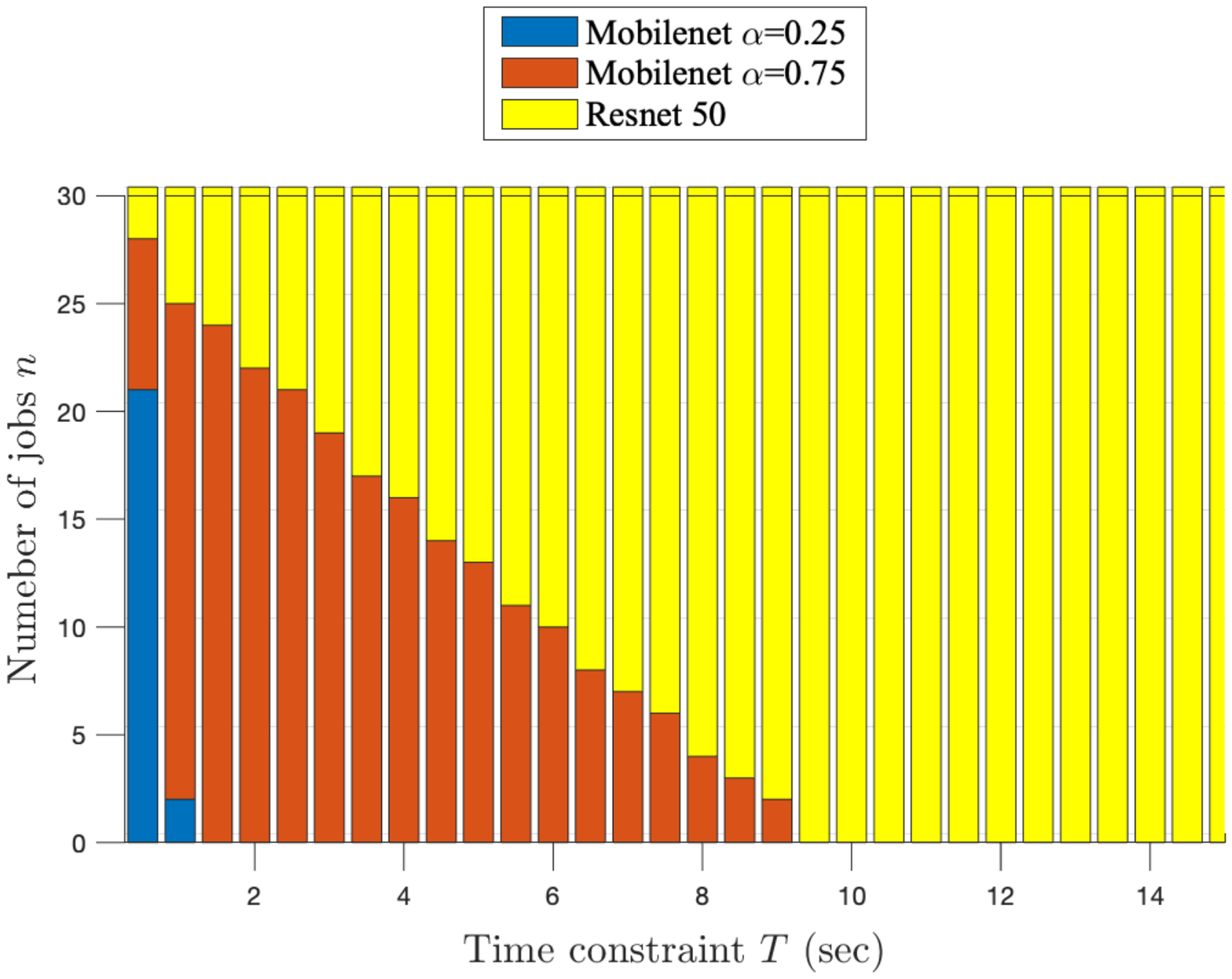}
\vspace{-4.6cm}
\caption{Job assignment under \algo for varying $T$.}
\label{fig:offloading_30jobs}
\end{figure}
In Figure~\ref{fig:offloading_30jobs}, we examine the number of jobs assigned to different models under \algo. Observer that as $T$ increases the number of jobs assigned to larger models increases. Also, note that MobileNet $\alpha = 0.25$ is only being used when $T$ is small. In all the subsequent figures, for each point, we run 30 experiments and compute the average. Recall that the total accuracy $A^\dagger$ is based on the average test accuracy of the models. In addition to $A^\dagger$, we also present the total true accuracy for \algo by summing the Top-1 accuracies we observe from executing the images under the given schedule $\mathbf{x}^\dagger$. In Figures \ref{fig:totalacc_fixedN_experimental} and \ref{fig:accuracy_fixedT_experimental}, we compare total accuracy achieved under different schedules, by varying $T$ and $n$, respectively. For $n= 60$, no LP-relaxed solution exists for $T = 0.5$ sec.

From both figures, we observe that $A^\dagger$ overlaps with, and in some cases exceeds, the total accuracy of the LP-relaxed solution $A^*_\text{LP}$. This is because all the processing times (cf. Table~\ref{tab:processing_times}) are less than $0.5$ sec, the minimum value used for $T$, and therefore, from Corollary~\ref{cor:accuracyBound} the worst-case bound is at most $a_{3} - a_1 = 0.376$ (cf. Table~\ref{tab:avg_accuracies}). Furthermore, in some cases, where $T$ is large enough, \algo may assign both the fractional jobs to the server and $A^\dagger$ exceeds $A^*_\text{LP}$. In these cases, however, the makespan under $\mathbf{x}^\dagger$ exceeds $T$.

In Figure \ref{fig:totalacc_fixedN_experimental} we observe that the total true accuracy of \algo is lower than $A^\dagger$, while in Figure~\ref{fig:accuracy_fixedT_experimental} it exceeds when the number of jobs are small. We note that this behaviour is highly dependent on the set of images we chose and the DNN models. This is expected because the true accuracy for an image on a model can have a large deviation from the average test accuracy of that model. From Figure~\ref{fig:totalacc_fixedN_experimental}, we observe that \algo always has higher total true accuracy than \blalgo with a percentage gain between $20$--$60$\% with an average of $40$\%, with lower percentage gains at smaller $T$. The latter fact is also confirmed in Figure~\ref{fig:accuracy_fixedT_experimental} when $T = 0.5$ sec. This is expected, because for $T = 0.5$ sec, not many jobs can be offloaded to the server as the processing times are around $0.3$ seconds. For $T = 4$ we again see significant gains of around $40$--$50$\%.

\begin{figure}[ht!]
\centering
%\vspace{-0.3cm}
\includegraphics[width=2.8in]{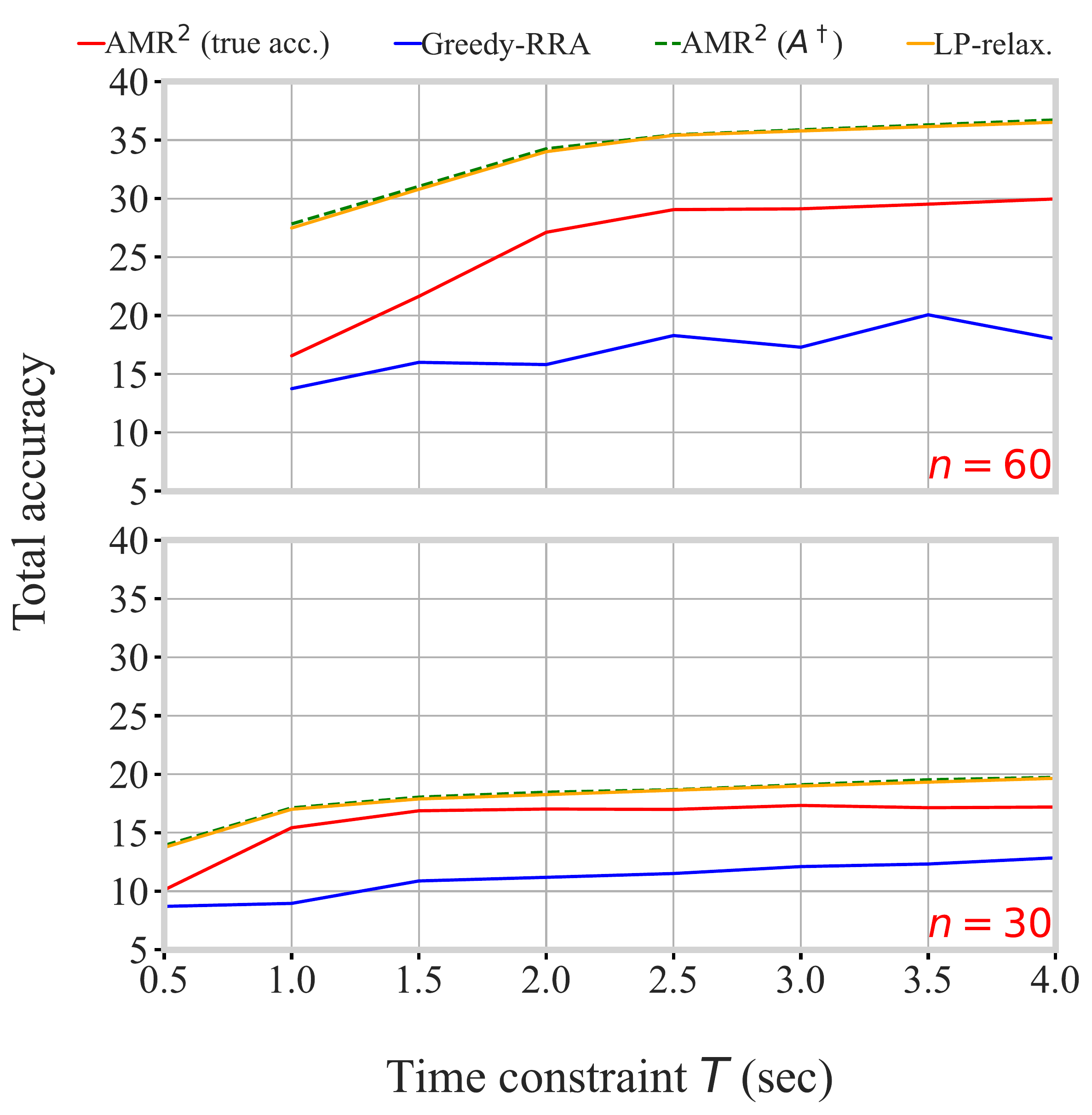}
\caption{Comparison of the total accuracy under different schedules for varying $T$ and for $n = 30$ and $n=60$.}
\label{fig:totalacc_fixedN_experimental}
\end{figure}

\begin{figure}[ht!]
\centering
\vspace{-0.3cm}
\includegraphics[width=2.8in]{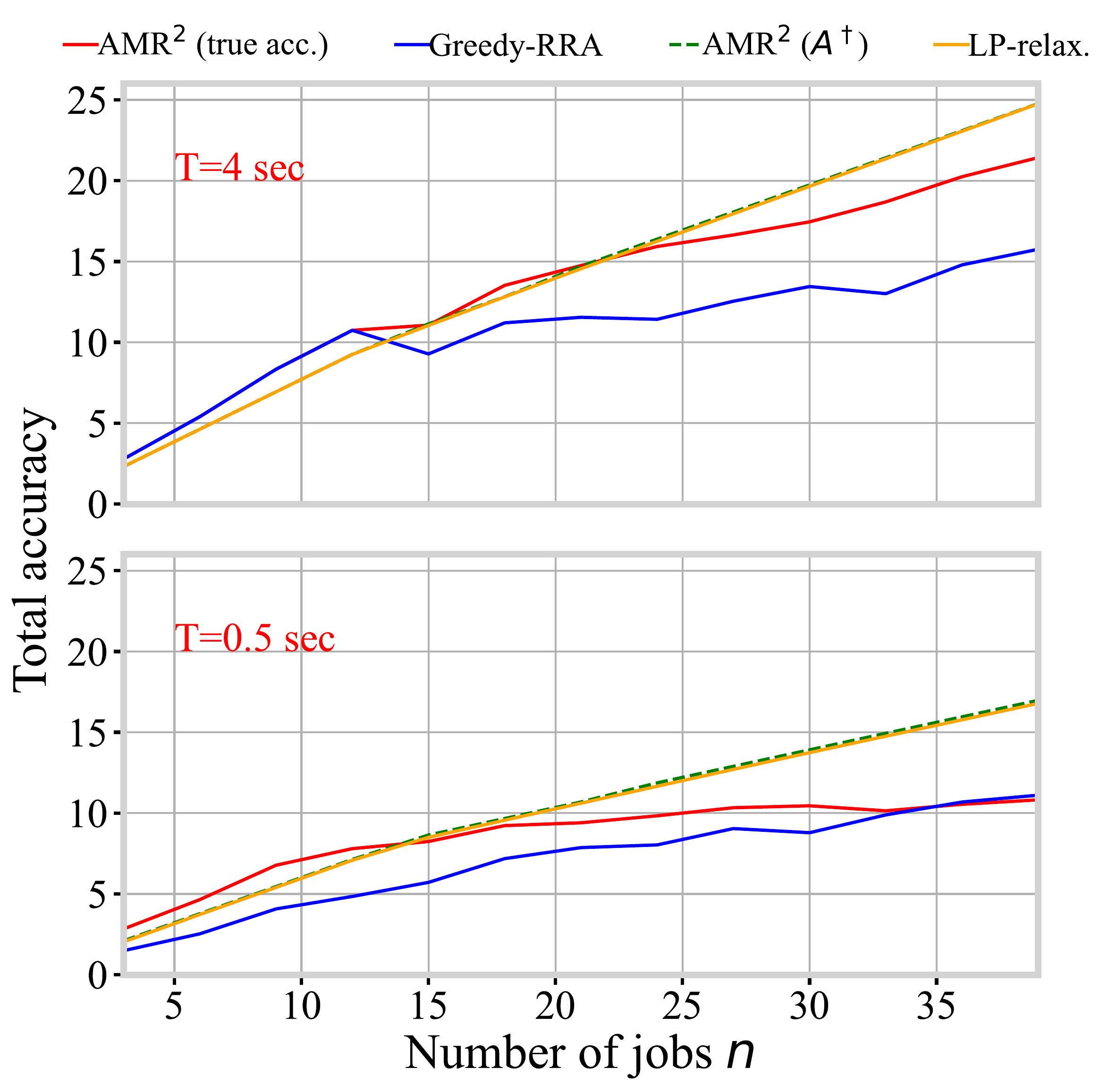}
\caption{Comparison of the total accuracy under different schedules for varying $n$ and $T=0.5$ sec and $T = 4$ sec.}
\label{fig:accuracy_fixedT_experimental}
\end{figure}

%The total accuracy when $N=60$ is higher than when $N=30$. This is because the number of the contribution to the summation for calculating the total accuracy A is higher in the case $N=60$. However, comparing the average accuracy per image when \algo solution has not arrived to steady state, we observe that when $N=30$ the average accuracy is major then when $N=60$. This is because, even if \algo offloads same number of images to the ES, the percentage of the batch set offloaded on Server is higher as batch size decreases. The steady state is when \algo offloads all images to ES. We can observe that when $N=30$ and $T \geq 1$, \algo is converging to the steady state. When $N=60$ and $T=1$ \algo has not reached the knee of the curve. 
%\begin{figure}[ht!]
%\includegraphics[width=2.8in]{figures/ExperimentalResults/makespan_N=(60,30).pdf}
%\caption{Comparison of makespan when $N=60$ or $N=30$}
%\label{fig:makespan_fixedN_experimental}
%\end{figure}

\begin{figure}[ht!]
\centering
\includegraphics[width=2.8in]{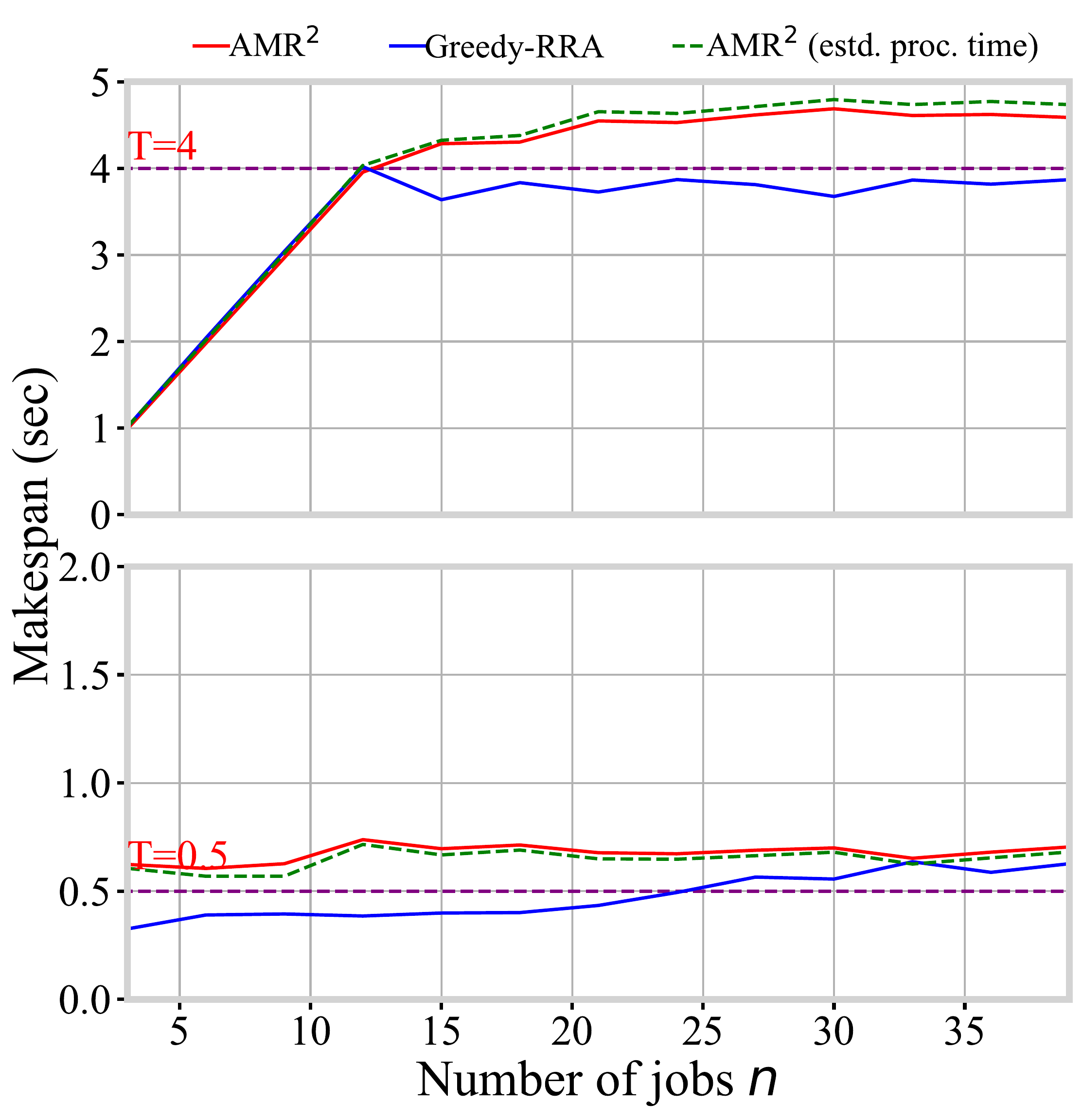}
\caption{Makespan under \algo and \blalgo for varying $n$, and $T=0.5$ sec and $T = 4$ sec.}
\vspace{-0.3cm}
\label{fig:makespan_fixedT_experimental}
\end{figure}
In Figure \ref{fig:makespan_fixedT_experimental}, we present the makespan achieved by \algo and \blalgo for varying $n$. The real-time makespan, i.e., the time elapsed at Raspberry Pi from the start of scheduling the jobs till the finishing time of the last job, is indicated by \algo in the legend. The estimated makespan that is numerically computed using the schedule $\mathbf{x}^\dagger$ and the estimated processing and communication times is indicated by \algo (estd. proc. time). Observe that both these makespans have negligible difference asserting that the variances in our estimates for both communication and processing times are small. Also, observe that both \algo and \blalgo violate the time constraint for different problem instances. For $T=4$, \algo violates $T$ for $n \geq 13$, but then it saturates at a makespan with maximum percentage of violation of $15$\%. This is expected, because from Lemma~\ref{lem:2fractionaljobs} there cannot be more than two fractional jobs irrespective of $n$ value and thus, the constraint violation due to the reassignment of the fractional jobs do not increase beyond $n = 30$. This saturation effect can also be observed for $T = 0.5$. In this case, the percentage of violation under \algo is higher, up to $40$\%, because the processing times on the server are comparable to $T = 0.5$ sec and reassigning a fractional job to the server results in higher percentage of violation.

\section{Conclusion}\label{sec:conclusion}
We have studied the offloading decision for inference jobs between an ED and an ES, where the ED has $m$ models and the ES has a state-of-the-art model. 
Given $n$ data samples at the ED, we proposed an approximation algorithm \algo for maximizing the total accuracy for the inference jobs subject to a time constraint $T$ on the makespan. We proved that the makespan under \algo is at most $2T$, and its total accuracy is lower than the optimal total accuracy by at most $2$, and for typical problem instances it is lower by at most $a_{m+1}-a_1$. When the data samples are identical, we have proposed AMDP, a pseudo-polynomial time algorithm to compute the optimal schedule. We have implemented \algo on Raspberry Pi and demonstrated its efficacy in improving the inference accuracy for classifying images within a time constraint $T$. Also, under the considered scenarios \algo provides, on average, $40$\% higher total accuracy than that of \blalgo. 

In our problem model, we considered that the communication times are deterministic and in our testbed we used an architecture where the ED and the ES are connected via Ethernet. However, communication times will be random if one considers offloading over wireless links, which we leave for future work.

%\clearpage
\bibliographystyle{IEEEtran}
\bibliography{main}

\begin{thebibliography}{10}
\providecommand{\url}[1]{#1}
\csname url@samestyle\endcsname
\providecommand{\newblock}{\relax}
\providecommand{\bibinfo}[2]{#2}
\providecommand{\BIBentrySTDinterwordspacing}{\spaceskip=0pt\relax}
\providecommand{\BIBentryALTinterwordstretchfactor}{4}
\providecommand{\BIBentryALTinterwordspacing}{\spaceskip=\fontdimen2\font plus
\BIBentryALTinterwordstretchfactor\fontdimen3\font minus
  \fontdimen4\font\relax}
\providecommand{\BIBforeignlanguage}[2]{{%
\expandafter\ifx\csname l@#1\endcsname\relax
\typeout{** WARNING: IEEEtran.bst: No hyphenation pattern has been}%
\typeout{** loaded for the language `#1'. Using the pattern for}%
\typeout{** the default language instead.}%
\else
\language=\csname l@#1\endcsname
\fi
#2}}
\providecommand{\BIBdecl}{\relax}
\BIBdecl

\bibitem{Shi2016}
W.~Shi, J.~Cao, Q.~Zhang, Y.~Li, and L.~Xu, ``Edge computing: Vision and
  challenges,'' \emph{IEEE Internet of Things Journal}, vol.~3, no.~5, pp.
  637--646, 2016.

\bibitem{Mach2017}
P.~Mach and Z.~Becvar, ``Mobile edge computing: A survey on architecture and
  computation offloading,'' \emph{IEEE Communications Surveys Tutorials},
  vol.~19, no.~3, pp. 1628--1656, 2017.

\bibitem{Deng2020}
L.~Deng, G.~Li, S.~Han, L.~Shi, and Y.~Xie, ``Model compression and hardware
  acceleration for neural networks: A comprehensive survey,'' \emph{Proceedings
  of the IEEE}, vol. 108, no.~4, pp. 485--532, 2020.

\bibitem{TFLite}
``Image classification using {TensorFlow Lite},''
  \url{https://www.tensorflow.org/lite/guide/hosted_models}.

\bibitem{PyTorch}
``Pytorch mobile,'' \url{https://pytorch.org/mobile/home/}.

\bibitem{He2016}
K.~He, X.~Zhang, S.~Ren, and J.~Sun, ``Deep residual learning for image
  recognition,'' in \emph{Proc. IEEE CVPR}, 2016, pp. 770--778.

\bibitem{Imagenet2009}
J.~Deng, W.~Dong, R.~Socher, L.-J. Li, K.~Li, and L.~Fei-Fei, ``Imagenet: A
  large-scale hierarchical image database,'' in \emph{Proc. IEEE CVPR}, 2009,
  pp. 248--255.

\bibitem{Sandler2018}
M.~Sandler, A.~Howard, M.~Zhu, A.~Zhmoginov, and L.-C. Chen, ``Mobilenetv2:
  Inverted residuals and linear bottlenecks,'' in \emph{Proc. IEEE CVPR}, 2018,
  pp. 4510--4520.

\bibitem{teerapittayanon2016}
S.~{Teerapittayanon}, B.~{McDanel}, and H.~{Kung}, ``Branchynet: Fast inference
  via early exiting from deep neural networks,'' in \emph{Proc. ICPR}, 2016,
  pp. 2464--2469.

\bibitem{Han2019}
H.~Cai, C.~Gan, and S.~Han, ``Once for all: Train one network and specialize it
  for efficient deployment,'' \emph{CoRR}, vol. abs/1908.09791, 2019.

\bibitem{Pinedo2008}
M.~L. Pinedo, \emph{Scheduling: Theory, Algorithms, and Systems}, 3rd~ed.\hskip
  1em plus 0.5em minus 0.4em\relax Springer Publishing Company, Incorporated,
  2008.

\bibitem{Kellerer2004}
H.~Kellerer, U.~Pferschy, and D.~Pisinger, \emph{Knapsack Problems}.\hskip 1em
  plus 0.5em minus 0.4em\relax Springer, Berlin, Germany, 2004.

\bibitem{Ross1975}
G.~Ross and R.~Soland, ``A branch and bound algorithm for the generalized
  assignment problem,'' \emph{Mathematical Programming}, vol.~8, pp. 91--103,
  1975.

\bibitem{Shmoys1993}
D.~Shmoys and E.~Tardos, ``An approximation algorithm for the generalized
  assignment problem,'' \emph{Mathematical Programming}, no.~62, pp. 461--474,
  1993.

\bibitem{Satyanarayanan2009}
M.~Satyanarayanan, P.~Bahl, R.~Caceres, and N.~Davies, ``The case for vm-based
  cloudlets in mobile computing,'' \emph{IEEE Pervasive Computing}, vol.~8,
  no.~4, pp. 14--23, 2009.

\bibitem{Liu2016}
J.~Liu, Y.~Mao, J.~Zhang, and K.~B. Letaief, ``Delay-optimal computation task
  scheduling for mobile-edge computing systems,'' in \emph{Proc. IEEE ISIT},
  2016, pp. 1451--1455.

\bibitem{Mao2016}
Y.~Mao, J.~Zhang, and K.~B. Letaief, ``Dynamic computation offloading for
  mobile-edge computing with energy harvesting devices,'' \emph{IEEE Journal on
  Selected Areas in Communications}, vol.~34, no.~12, pp. 3590--3605, 2016.

\bibitem{Champati2017}
J.~P. Champati and B.~Liang, ``Semi-online algorithms for computational task
  offloading with communication delay,'' \emph{IEEE Transactions on Parallel
  and Distributed Systems}, vol.~28, no.~4, pp. 1189--1201, 2017.

\bibitem{Chen2015}
M.-H. Chen, B.~Liang, and M.~Dong, ``A semidefinite relaxation approach to
  mobile cloud offloading with computing access point,'' in \emph{Proc. IEEE
  SPAWC Worskhop}, 2015, pp. 186--190.

\bibitem{Kamoun2015}
M.~Kamoun, W.~Labidi, and M.~Sarkiss, ``Joint resource allocation and
  offloading strategies in cloud enabled cellular networks,'' in \emph{Proc.
  IEEE ICC}, 2015, pp. 5529--5534.

\bibitem{Wang2016}
Y.~Wang, M.~Sheng, X.~Wang, L.~Wang, and J.~Li, ``Mobile-edge computing:
  Partial computation offloading using dynamic voltage scaling,'' \emph{IEEE
  Transactions on Communications}, vol.~64, no.~10, pp. 4268--4282, 2016.

\bibitem{Wang2019}
Z.~Wang, W.~Bao, D.~Yuan, L.~Ge, N.~H. Tran, and A.~Y. Zomaya, ``See:
  Scheduling early exit for mobile dnn inference during service outage,'' in
  \emph{in Proc. MSWIM}, 2019, p. 279–288.

\bibitem{Ogden2020}
S.~S. Ogden and T.~Guo, ``Mdinference: Balancing inference accuracy and latency
  for mobile applications,'' in \emph{Proc. IEEE IC2E}, 2020, pp. 28--39.

\bibitem{Nikoloska2021}
I.~Nikoloska and N.~Zlatanov, ``Data selection scheme for energy efficient
  supervised learning at iot nodes,'' \emph{IEEE Communications Letters},
  vol.~25, no.~3, pp. 859--863, 2021.

\bibitem{Krzyszof1989}
K.~Dudzinski, ``On a cardinality constrained linear programming knapsack
  problem,'' \emph{Operations Research Letters}, vol.~8, no.~4, pp. 215--218,
  1989.

\bibitem{Cattrysse1992}
D.~G. Cattrysse and L.~N. {Van Wassenhove}, ``A survey of algorithms for the
  generalized assignment problem,'' \emph{European Journal of Operational
  Research}, vol.~60, no.~3, pp. 260--272, 1992.

\bibitem{Chekuri2000}
C.~Chekuri and S.~Khanna, ``A ptas for the multiple knapsack problem,'' in
  \emph{Proceedings of the Eleventh Annual ACM-SIAM Symposium on Discrete
  Algorithms}, ser. SODA '00.\hskip 1em plus 0.5em minus 0.4em\relax USA:
  Society for Industrial and Applied Mathematics, 2000, p. 213–222.

\bibitem{Ross2019}
P.~Ross and A.~Luckow, ``Edgeinsight: Characterizing and modeling the
  performance of machine learning inference on the edge and cloud,'' in
  \emph{2019 IEEE International Conference on Big Data (Big Data)}, 2019, pp.
  1897--1906.

\bibitem{Potts1984}
C.~Potts, ``Analysis of a linear programming heuristic for scheduling unrelated
  parallel machines,'' \emph{Discrete Applied Mathematics}, vol.~10, no.~2, pp.
  155--164, 1985.

\bibitem{Jan2020}
J.~van~den Brand, ``A deterministic linear program solver in current matrix
  multiplication time,'' in \emph{Proc. ACM SODA}.\hskip 1em plus 0.5em minus
  0.4em\relax Society for Industrial and Applied Mathematics, 2020, p.
  259–278.

\bibitem{Howard2017}
A.~G. Howard, M.~Zhu, B.~Chen, D.~Kalenichenko, W.~Wang, T.~Weyand,
  M.~Andreetto, and H.~Adam, ``Mobilenets: Efficient convolutional neural
  networks for mobile vision applications,'' \emph{CoRR}, vol. abs/1704.04861,
  2017.

\end{thebibliography}

\end{document}